\documentclass[10pt,twoside]{article}
\usepackage[dvips]{graphicx}
\usepackage {amsmath, amssymb, amsthm, mathrsfs, amsfonts, bbm, stmaryrd, textcomp} 
\usepackage[all, 2cell,matrix,arrow,curve,cmtip]{xy}
\usepackage[abs]{overpic}

\usepackage{fancyhdr}
\renewcommand{\subsectionmark}[1]{}
\newlength\Li \newlength\Lii
\newcommand{\vdim}{\mbox{{dim}}}

\newcommand{\esssup}{\mathop{\mathrm{ess\,sup}}}

\date{}

\makeatletter
\def\@seccntformat#1{}
\makeatother
\numberwithin{equation}{section}

\makeatletter
\def\blfootnote{\xdef\@thefnmark{}\@footnotetext}
\makeatother

\title{Existence and Uniqueness of Strong Solutions for a Compressible Multiphase Navier-Stokes Miscible Fluid-Flow Problem in Dimension $n=1$} 
\author{C.~Michoski\textsuperscript{\textdagger} \ \& \ A.~Vasseur\textsuperscript{*} \\ \\ \small{Departments of Mathematics, Chemistry and Biochemistry} \\ \small{University of Texas at Austin}}
\begin{document}\maketitle\begin{abstract}
We prove the global existence and uniqueness of strong solutions for
a compressible multifluid described by the barotropic Navier-Stokes equations in $\vdim =1$. The result holds when the diffusion coefficient depends on the pressure. It relies on a global control in time of the $L^2$ norm of the space derivative of the density, via a new kind of entropy.
\end{abstract}

\blfootnote{\textsuperscript\textdagger {\it michoski@cm.utexas.edu}, Department of Chemistry and Biochemistry}\blfootnote{*{\it vasseur@math.utexas.edu}, Department of Mathematics}

\section{\protect\centering $\S 1$ Introduction}

In this paper we show the well-posedness of a global strong solution to a multifluid problem over $\mathbb{R}^{+}\times\mathbb{R}$ characterized by the one-dimensional compressible barotropic Navier-Stokes equations.  That is, we consider the following system of equations, \begin{align} \label{rrr}\partial_{t}\rho+\partial_{x}&(\rho u)=0,\\\label{sss}\partial_{t}(\rho u)+\partial_{x}(\rho u^{2})+\partial_{x}p(\rho,&\mu )-\partial_{x}(\nu(\rho,\mu)\partial_{x}u)=0, \\ \label{ttt}\partial_{t}(\rho\mu)+\partial_{x}&(\rho u\mu)=0, \end{align} with initial conditions given by: $$\rho_{|t=0}=\rho_{0}>0,\qquad\rho u_{|t=0}=m_{0},\qquad\mu_{|t=0}=\mu_{0}.$$  The conservation of mass (\ref{rrr}), conservation of momentum (\ref{sss}) and conservation of species (\ref{ttt}) describe the flow of a barotropic compressible viscous fluid defined for $(t,x)\in \mathbb{R}^{+}\times\mathbb{R}$.  Here the \emph{density} is given as $\rho$, the \emph{velocity} as $u$, the \emph{momentum} as $m$, and the \emph{mass fraction} $\mu$ denotes the relative weighting for each fluid component of the \emph{adiabatic exponent} $\gamma(\mu)\in \mathbb{R}$ associated to the generalized \emph{pressure} $p(\rho,\mu)$, thus effectively tracking the ``mixing'' of the fluid components.  In one dimension the shear viscosity and the bulk viscosity collapse into a single coefficient function depending on $\rho$ and $\mu$ which we denote here $\nu(\rho,\mu)$.     
Monofluid one-dimensional compressible Navier-Stokes equations have been studied by many authors when the viscosity coefficient $\nu$ is a positive constant. The existence of weak solutions was first established by A.~Kazhikhov and V.~Shelukhin \cite{KS} for smooth enough data close to the equilibrium (bounded away from zero). The case of discontinuous data (still bounded away from zero) was addressed by V.~Shelukhin \cite{Sh1,Sh2} and then by D.~Serre \cite{Serre1,Serre2} and D.~Hoff \cite{Hoff1}. First results concerning vanishing intial density were also obtained by V.~Shelukhin \cite{Sh4}.  In \cite{Hoff4}, D.~Hoff proved the existence of global weak solutions with large discontinuous initial data, possibly having different limits at $x=\pm\infty$. He proved moreover that the constructed solutions have strictly positive densities (vacuum states cannot form in finite time). In dimension greater than two, similar results were obtained by A.~Matsumura and T.~Nishida \cite{MN1,MN2,MN3} for smooth data and D.~Hoff \cite{Hoff2} for discontinuous data close to the equilibrium. The first global existence result for initial density that are allowed to vanish was due to P.-L.~Lions (see \cite{PLL2}).  The result was later improved by B.~Desjardins (\cite{BD5}) and E.~Feireisl et al. (\cite{CHT, F1, F2, F3}). The class of solutions was then extended by A.~Zlotnik, G.-Q.~Chen, D.~Hoff, B.~Ducomet, and K.~Trivisa in \cite{CHT2, CT, DZ, KT, Z1, ZD} to the case of a thermally active compressible flow coupled by the systems chemical kinetics, where global existence results are shown for an Arrhenius type biphasic combustion reaction tracking only the reactants level of consumption.  Y.~Amirat and V.~Shelukin have further provided \cite{AS} weak solutions for the case of a miscible flow in porous media.

The problem of regularity and uniqueness  of solutions was first analyzed by V.~Solonnikov \cite{S1} for smooth initial data and for small time. However, the regularity may blow-up as the solution gets close to vacuum. This leads to another interesting question of whether vacuum may arise in finite time. D. Hoff and J.~Smoller (\cite{HS}) show that any weak solution of the Navier-Stokes equations in one space dimension do not exhibit vacuum states, provided that no vacuum states are present initially.

Interfacial multicomponent flows have been extensively studied in the literature, and span a rich array of applied topics with natural analogues in continuum dynamics.  For example, there has been numerous work on multicomponent flows in biological systems, including bifurcating vascular flows \cite{DEP, FW} and pulsatile hemodynamics \cite{JLPH}, \emph{in vitro} tissue growth \cite{LKBJS} and the ``am\oe boid motion'' of cells by way of surface polymerization \cite{Evans}, chemotactic transport \cite{F5} in aqueous media under chemical mixing (e.g. varying relative concentrations) applied to specialized cell types \cite{Eisenbach}, as well as biological membrane dynamics due to local gradients in surface tension caused by flux in local boundary densities \cite{SIG}.  Another important and popular field of application is that of dispersed nanoparticles in colloidal media (e.g. aerosols, emulsifications, sols, foams, etc.) \cite{Mayer, MV2, TGS} applied to, for example, electrospray techniques in designing solar cells \cite{FTKH}, or more generally to diagnostic and flow analysis in the applied material sciences \cite{BMR, SEV}.  In addition, phase separation and spinodal decomposition have received a great deal of attention \cite{Binder, CH, LWZ}, especially with respect to morphological engineering \cite{HJ}.   Another field which is heavily weighted with multifluid applications is that of combustion dynamics \cite{Williams} and chemical kinetics \cite{Gardiner, Houston}, where the conservation of species equation (\ref{ttt}) is regularly invoked, including numerous topics in reaction diffusion dynamics and phase mixing, spanning many essential topics in the atmospheric \cite{DS, HG} and geophysical \cite{Pedlosky, Vallis} sciences.  In electrochemistry and chemical engineering recent work has been done on porous multiphase fuel cells \cite{PSW, PW}, and in sonochemistry recent studies have shown acoustically induced transport properties across interfacial phase changes \cite{DN, MS}.  Finally in the fields of astronomy and astrophysics exotic multicomponent magnetohydrodynamic plasmas are studied \cite{MRFJ, OD}.       

Many applied results exist for computational methods and schemes for solving multicomponent flows.  Let us briefly mention some notable examples.  An early generalized numerical approach in multiphase modeling was presented by F.~Harlow and A.~Amsden in \cite{HA}, which provides an extensive system of dynamically coupled phases using a conservation of species (\ref{ttt}) equation obeying a number of relevant physical boundary conditions and which applies, in particular, to compressible flows.  In \cite{Dukowicz} J.~Dukowicz implements a particle-fluid model for incompressible sprays, an approach extended by G.~Faeth in \cite{Faeth, Faeth2} to combustion flows.  D.~Youngs then, in \cite{Youngs}, expanded numerical mixing regimes to include interfacial turbulent effects.  These basic schemes and approaches have been applied by a large number of authors to a large number of fields, modeling an extremely diverse number of natural phenomenon, from star formation \cite{HSWK} to volcanic eruptions \cite{OCENS}.  Some good reviews of the foundational numerics of these approaches can be found in the books of C.~Hirsch \cite{Hirsch}, P.~Shih-I and L. Shijun \cite{PL} and M.~Feistauer, J.~Felcman, and I.~Stra{\v{s}}kraba \cite{FFS}.       

Let us briefly outline the physical meaning of the subject of this paper, namely, the system of equations (\ref{rrr})-(\ref{ttt}).  Here we have a barotropic system with the flow driven by a pressure $p$ that depends on the density $\rho$ and the mass fraction $\mu$ of each chemical/phase component of the system.  Since the function $\gamma(\mu)$ depends on the constant heat capacity ratios $\gamma_{i}>1$ of each component of the multifluid, the pressure $p(\rho,\mu)$ effectively traces the thermodynamic ``signature'' of mixing chemicals/phases in solution.  Note that this is very similar, for example, to the system of equations set out in \cite{Williams}, except here, for simplicity, we have neglected the associated diffusion and chemical kinetics which break the strict (and mathematically convenient) conservation in the species equation (\ref{ttt}).  Another important facet of the system (\ref{rrr})-(\ref{ttt}) is that the viscosity $\nu$ is a function of the pressure $p$.  Much recent work has been done by M.~Franta, M.~Bul{\'{\i}}{\v{c}}ek, J.~M{\'a}lek, and K.~Rajagopal in providing results on these type of viscosity laws \cite{BMR2, FMR, MR1, MR2}.  Moreover, since the form of the pressure $p$ is chosen up to any state equation that satisfies the assumptions given in \textsection{2}, the formulation is general enough to include, for example, multi-nuclear regimes.  That is, in addition to describing the flow of mixing fluids characterized by their concentrations with respect to their heat capacity ratios, this construction also educes applications in nuclear hydrodynamics, where one can derive the pressure law using the time-dependent \emph{Hartee-Fock approximation} \cite{WMW}.  Such a \emph{nuclear fluid} obeying the assumptions given by the \emph{Eddington Standard Model} for stellar phenomena has a pressure law \cite{F1} that takes the form, $p(\rho)=C_{1}\rho^{3}-C_{2}\rho^{2}+C_{3}\rho^{7/4}$ where $C_{1},C_{2}$ and $C_{3}$ are positive constants.   In particular, this exotic pressure law can be shown to model nontrivial physical phenomena; such as spin and isospin wavefront propogation in nuclear fluids.  It has further been shown to be in good agreement with nuclear hydrodynamic models of the sun \cite{F1}.  Thus the result in \textsection{2} allows us to extend the above to \emph{nuclear multifluids} that satisfies $$p(\rho,\mu)= C_{1}\rho^{\gamma_{i}(\mu)}-C_{2}\rho^{\gamma_{j}(\mu)}+C_{3}\rho^{\gamma_{k}(\mu)},$$ as long as it verifies the conditions given in \textsection{2}.  It however remains to be seen if quantum multi-molecular fluids \cite{Wyatt} have an analogous formulation. 

At the level of the mathematical results, incompressible (for general background on the incompressible Navier-Stokes equations see \cite{PLL1}) multicomponent flows have been addressed by a number of authors.  First S.~Antontsev and A.~Kazhikhov in \cite{AK}, A.~Kazhikhov in \cite{Kazhikhov}, S.~Antontsev, A.~Kazhikhov and V.~Monakhov in \cite{AKM}, and B.~Desjardins in \cite{BD5} show results for mixing flows where homogenization of the density $\rho$ is allowed.  These solutions can be seen in contrast with P.L.~Lions' and R.~DiPerna's solutions in \cite{DL, NP1} which provide a multiphase solution for immiscible inhomogenizable flows given discrete constant densities for each component.  A.~Nouri, F.~Poupaud and Y.~Demay \cite{NP1, NP2} extend these results to functional densities where boundary components $\partial\Omega_{i}$ are set between each fluid domain $\Omega_{i}$ that satisfy the so-called \emph{kinematic condition}, which restricts the viscosity $\nu$ to obey $\partial_{t}\nu+u\cdot\nabla\nu=0$ (see \cite{NS} for further discussion on the kinematic condition).  These results apply to immiscible flows with boundary surfaces that effectively fix the number of fluid particles on the interface.  These solutions were then further extended by N.~Tanaka \cite{T1,T2}, V.~Solonnikov and A.~Tani \cite{S2, ST2, S3, S4, S5} to include boundary conditions tracking both the surface tension at the interface using a mean curvature flow on the interfacial surface, as well as the inclusion of self-gravitating parcels.

In this paper we consider viscosity coefficients depending on the pressure satisfying a barotropic-type pressure law, a result based upon the paper of A.~Mellet and A.~Vasseur in \cite{MV1} and extended to the multifluid case with a viscosity functional $\nu(p)$ given no \emph{a priori} uniform bound from below.  Thus, in addition to modeling the miscible multiflow regimes that have generated substantial physical interest (see above), our result further incorporates a very inclusive form of the generalized viscosity.  We show the global existence with uniqueness result for a one-dimensional compressible barotropic multicomponent Navier-Stokes problem.  In order to aquire the existence result, we rely heavily on an energy inequality provided by D.~Bresch and B.~Desjardins (see for example \cite{BD0} and \cite{BD2}).  This beautiful and powerful tool is central to our result, and, as it turns out, the breakdown of this calculation is the only (known) obstruction to acquiring similar results in dimension greater than one.  Next we obtain the uniqueness result by adapting a proof of Solonnikov's \cite{S1} to the case of the barotropic system (\ref{rrr})-(\ref{sss}) coupled to the species conservation equation (\ref{ttt}). 

Let us take this opportunity to discuss difficulties and related systems of equations in higher dimension.  Again, the present result relies heavily on the calculation of an energy inequality (see \textsection{3}) as provided by D.~Bresch and B.~Desjardins \cite{BD4, BD1, BDL}.  However, in dimension $n\geq 2$, the derivation of this entropy inequality leads to an unnatural form of the viscosity coefficiant $\nu(\rho,\mu)$; which is to say, the calculation no longer demonstrates the type of symmetry which leads to the essential \emph{cancellation} of \emph{singularities} (for example see \cite{Hoff4}) required in the calculation (see \cite{MV3} for the monofluid case). 

We briefly recall some exciting results known for compressible fluids in higher dimension, and note that extending these to multifluid regimes introduces both beautiful and difficult mathematics, while also addressing very important and physically relevant questions in the applied fields.  For example, a result of A.~Valli and W.~Zajaczkowski \cite{VZ} shows global weak solutions to the mutlidimensional problem for a heat conducting fluid with inflow and outflow conditions on the boundary.  In \cite{ST} A.~Solonnikov and A.~Tani offer a uniqueness proof for an isentropic compressible problem given a free boundary in the presence of surface tension. D.~Hoff and E.~Tsyganov next provide a very nice extension of the system to find weak solutions to the compressible magnetohydrodynamics regime in \cite{HT}.  G.~Chen and M.~Kratka in \cite{ChenK} further show a free boundary result for a heat-conducting flow given spherically symmetric initial data and a constant viscosity coefficient in higher dimension; a result which is extended by E.~Feireisl's work \cite{F1} under the notion of the \emph{variational solution} for heat-conducting flows in multiple dimensions; though this result restricts the form of the equation of state.  Further existence results are provided by B.~Ducomet and E.~Feireisl \cite{DF, DF2} for gaseous stars and the compressible heat-conducting magnetohydrodynamic regime.  Further, D.~Donatelli and K.~Trivisa \cite{DT2, DT} have extended the existence results for the coupled chemical kinetics system mentioned above to higher dimensions.  In \cite{MV2} A.~Mellet and A.~Vasseur provide global weak solutions for a compressible barotropic regime coupled to the Vlasov-Fokker-Planck equation, which characterizes the evolution of dispersed particles in compressible fluids, such as with spray phenomenon.  Finally, important results of D.~Bresch and B.~Desjardins, in a very recent paper \cite{BD3}, has worked to extend the existence results to a more general framework (using their energy inequality) for a viscous compressible heat-conducting fluid.

We conclude by noting a number of important and interesting results related to vacuum solutions.  That is, though in this work we are concerned with densities that obey uniform bounds in $\mathbb{R}$, a number of nice results exist for the case where over some open $U\subset\mathbb{R}$, $$\int_{U}\rho_{0}dx\geq 0;$$ which is to say, solutions that incorporate vacuum states.  For example, T.~Yang and C.~Zhu show in \cite{YZ} global existence for a 1D isentropic fluid connected continuiously to a vacuum state boundary with a density dependent viscosity.  Additionally, in dimension one, a recent result by C.~Cho and H.~Kim \cite{CK} provides unique strong local solutions to a viscous polytropic fluid, where they utilize a compatibility condition on the initial data.

\section{\protect\centering $\S 2$ Statement of Result}

Let us first state the hypothesis we make on the pressure and viscosity functional $p(\rho,\mu)$ and $\nu(\rho,\mu)$. First we assume that the pressure $p(\rho,\mu)$ is an increasing function of the density $\rho$ such that a.e., \begin{equation}\label{rhoincrease}\partial_{\rho}p(\rho,\mu)\geq 0.\end{equation}

The viscosity coefficient  $\nu(\rho,\mu)$ is chosen such that it satisfies the following relation, \begin{equation}\label{visc}\nu(\rho,\mu)=\rho\partial_{\rho}p(\rho,\mu)\psi'(p(\rho,\mu)),\end{equation} where $\psi(p)$ is a function of the pressure restricted only by the form of its derivative in $p$.

We consider a multifluid for which the pressure functional does not change too much with respect to the fractional mass. Namely, Consider two $\check{\gamma}>1$ and $\hat{\gamma}>1$, where $\check{\gamma}<\gamma<\hat{\gamma}$ up to the constraint that, \begin{align}\label{gamma1} & \ \ \ \frac{\hat{\gamma}-1/2}{\check{\gamma}}<\frac{\check{\gamma}+1/2}{\hat{\gamma}},\\ \label{gamma2}&\frac{\check{\gamma}-1/2}{\hat{\gamma}}>\frac{\hat{\gamma}+1/2}{\check{\gamma}}-1.\end{align}  These relations are satisfied when $\hat{\gamma}=1.4$ and $\check{\gamma}=1.3$, for example. 

Then, we ascribe the existence of constants $C\geq 0$ such that the following conditions hold: \begin{equation}\begin{aligned}\label{rho2}\psi'(p)\geq C&\sup(p^{-\underline{\alpha}},p^{-\overline{\alpha}}),\\\rho^{\check{\gamma}}/C\leq p(\rho,\mu)\leq C&\rho^{\hat{\gamma}}\quad\mathrm{for}\quad\rho \geq 1,\mu\in\mathbb{R},\\ \rho^{\hat{\gamma}}/C\leq p(\rho,\mu)\leq C&\rho^{\check{\gamma}}\quad\mathrm{for}\quad\rho\leq 1,\mu\in\mathbb{R},\end{aligned}\end{equation} where $\underline{\alpha}$ and $\overline{\alpha}$ are such that \begin{align}\label{gamma3} & \ \ \ \frac{\hat{\gamma}-1/2}{\check{\gamma}}<\underline{\alpha}\leq\frac{\check{\gamma}+1/2}{\hat{\gamma}},\\ \label{gamma4}&\frac{\check{\gamma}-1/2}{\hat{\gamma}}>\overline{\alpha}\geq\frac{\hat{\gamma}+1/2}{\check{\gamma}}-1.\end{align}  Note that the existence of $\underline{\alpha}$ and $\overline{\alpha}$ comes from (\ref{gamma1}) and (\ref{gamma2}).

Next we set conditions on the derivatives of the pressure in $\rho$ and $\mu$; given first in $\rho$ by, \begin{equation}\begin{aligned}\label{rhopb}\rho^{\check{\gamma}-1}/C& \leq \partial_{\rho}p(\rho,\mu) \leq C\rho^{\hat{\gamma} -1} \quad \mathrm{for} \quad \rho\geq 1,\mu\in\mathbb{R},\\ \rho^{\hat{\gamma}-1}/C &\leq\partial_{\rho}p(\rho,\mu)\leq C\rho^{\check{\gamma} -1}\quad\mathrm{for}\quad\rho\leq 1,\mu\in\mathbb{R},\end{aligned}\end{equation} and in $\mu$ by, \begin{equation}\begin{aligned}\label{rhomub}&\partial_{\mu}p(\rho,\mu) \leq C\rho^{\hat{\gamma}} \quad\mathrm{for} \quad \rho\geq 1,\mu\in\mathbb{R},\\ &\partial_{\mu}p(\rho,\mu) \leq C\rho^{\check{\gamma}} \quad\mathrm{for} \quad \rho\leq 1,\mu\in\mathbb{R}.\end{aligned}\end{equation}  Notice that a simple pressure which satisfies these conditions is, $p(\rho,\mu)=C(\mu)\rho^{\gamma(\mu)}$ where $1/C\leq C(\mu)\leq C$ and with two constants $\gamma_{1}$ and $\gamma_{2}$ such that, $$\check{\gamma}<\gamma_{1}\leq \gamma(\mu) \leq \gamma_{2} < \hat{\gamma}.$$  In particular note that (\ref{rhoincrease}), (\ref{visc}), (\ref{rho2}), and (\ref{rhopb})-(\ref{rhomub}) are quite general assumptions, while the strong conditions, (\ref{gamma1}) and (\ref{gamma2}), have the effect of constraining the amount $p(\rho,\mu)$ can change with respect to $\mu$.

For the sake of clarity we define $\dot{H}^{1}(\mathbb{R})$ as the space consisting of all functions $\rho$ for which, $$\int_{\mathbb{R}}(\partial_{x}\rho)^{2}dx\leq C.$$  This paper is dedicated to the proof of the following theorem.

\newtheorem*{thpre}{Theorem}
\begin{thpre} Assume a pressure $p(\rho,\mu)$ and viscosity $\nu(\rho,\mu)$ satisfying (\ref{visc}), (\ref{rho2}), and (\ref{rhopb})-(\ref{rhomub}) where the adiabatic limits $\hat{\gamma}$ and $\check{\gamma}$ verify (\ref{gamma1})-(\ref{gamma2}), and take initial data $(\rho_{0}, u_{0},\mu_{0})$ for which there exists positive constants $\overline{\varrho}(0)$ and $\underline{\varrho}(0)$ such that $$\begin{aligned}  0 < \underline{\varrho}(0)\leq & \rho_{0} \leq \overline{\varrho}(0)< \infty, \\ \rho_{0}\in\dot{H}^{1}(\mathbb{R}), \quad u_{0} \in & H^{1}(\mathbb{R}), \quad \mu_{0} \in H^{1}(\mathbb{R}), \\ \int_{\Omega}\mathscr{E}(\rho_{0},& \mu_{0})dx< +\infty, \\ |\partial_{x}&\mu_{0}|\leq C\rho_{0},\end{aligned}$$ where $\mathscr{E}$ is the internal energy as defined in (\ref{genenerg}).  We additionally assume the existence of constants $R,S>0$ and $\tilde{\rho},\tilde{\mu}>0$ where $\rho_{0}\equiv\tilde{\rho}$ for $|x|>R$ and $\mu_{0}\equiv\tilde{\mu}$ for $|x|>S$.  Then there exists a global strong solution to (\ref{rrr})-(\ref{ttt}) on $\mathbb{R}^{+}\times\mathbb{R}$ such that for every $T>0$ we have $$\begin{aligned} \rho \in L^{\infty}(0,T; \dot{H}^{1}(\mathbb{R})), & \quad \partial_{t}\rho \in L^{2}((0,T)\times \mathbb{R}), \\ u \in L^{\infty}(0,T; H^{1} (\mathbb{R}))\cap L^{2} (0,T & ;H^{2}(\mathbb{R})), \ \partial_{t}u \in L^{2}((0,T)\times\mathbb{R}), \\ \mu_{x} \in L^{\infty}(0,T; L^{\infty}(\mathbb{R})), & \quad \partial_{t}\mu \in L^{\infty}(0,T;L^{2}(\mathbb{R})).\end{aligned}$$  Furthermore, there exist positive constants $\underline{\varrho}(T)$ and $\overline{\varrho}(T)$ depending only on $T$, such that $$0< \underline{\varrho}(T) \leq \rho(t,x) \leq \overline{\varrho}(T)< \infty, \quad\quad \forall (t,x)\in (0,T)\times\mathbb{R}.$$  Additionally, when $\psi''(p)$, $\partial_{\rho\rho}p(\rho,\mu)$, and $\partial_{\rho\mu}p(\rho,\mu)$ are each locally bounded then this solution is unique in the class of weak solutions satisfying the entropy inequalities of \textsection{3}. 
\end{thpre}

It is worth remarking that our results are slightly stronger than those presented in the statement of the theorem above.  Namely, the conditions $\rho_{0}\equiv\tilde{\rho}$ for $|x|>R$ and $\mu_{0}\equiv\tilde{\mu}$ for $|x|>S$ with respect to constants $R$ and $S$ can be relaxed, such that simply choosing $\rho_{0}$ and $\mu_{0}$ close to the reference values $\tilde{\rho}$ and $\tilde{\mu}$ is permissible as long as the internal energy $\mathscr{E}(\rho,\mu)$ remains integrable at $t=0$.  

We additionally use the existence of the short-time solution to the system (\ref{rrr})-(\ref{ttt}), which follows from \cite{S1}.  That is, as we show explicitly in \textsection{4}, applying (\ref{rho2}) and (\ref{rhopb}) to (\ref{visc}) provides that for every $(\rho,\mu)$ the viscosity coefficient $\nu(\rho,\mu)\geq C$, for a positive constant $C$.  This leads to the following proposition.    

\newtheorem*{solon}{Proposition 2.1 ({\it Solonnikov})}
\begin{solon} For initial data $(\rho_{0}, u_{0}, \mu_{0})$ taken with respect to the positive constants $\overline{\varrho}(0)$ and $\underline{\varrho}(0)$ satisfying $$\begin{aligned} 0 <\underline{\varrho}(0)\leq & \rho_{0} \leq\overline{\varrho}(0) < \infty, \\ \rho_{0} \in \dot{H}^{1}(\mathbb{R}), \quad u_{0} \in & H^{1}(\mathbb{R}), \quad \mu_{0} \in H^{1}(\mathbb{R}), \end{aligned}$$ and assuming that $\nu(\rho,\mu)\geq C$ for a positive constant $C$, then there exists a $T_{s}>0$ for each such that (\ref{rrr})-(\ref{ttt}) has a unique solution $(\rho,u,\mu)$ on $(0,T_{s})$ for each $T_{r}<T_{s}$ satisfying $$\begin{aligned} & \ \rho\in  L^{\infty}(0,T_{r};\dot{H}^{1}(\mathbb{R})), \quad\quad \partial_{t}\rho \in L^{2}((0,T_{r})\times \mathbb{R}), \\ & \ u\in  L^{2}(0,T_{r};H^{2}(\mathbb{R})), \quad\quad \ \partial_{t}u\in L^{2}((0,T_{r})\times\mathbb{R}), \\  & \mu_{x} \in L^{\infty}(0,T_{r};L^{\infty}(\mathbb{R})), \quad \ \partial_{t}\mu\in L^{\infty}((0,T_{r});L^{2}(\mathbb{R}));\end{aligned}$$ and there exists two positive constants, $\underline{\varrho}_{r}>0$ and $\overline{\varrho}_{r}<\infty$, such that $\underline{\varrho}_{r}\leq \rho(x,t)\leq \overline{\varrho}_{r}$ for all $t\in (0,T_{s})$.
\end{solon}

The proof of Solonnikov's proposition 2.1 as presented in \cite{S1} follows with the addition of equation (\ref{ttt}) by applying Duhamel's principle to the transport equation in $\mu$ given the regularity which we demonstrate in \textsection{4}, in much the same way Duhamel's principle is applied to $\rho$ in \cite{S1} for the continuity equation.  The rest of the proof then pushes through directly by virtue of the calculation shown in \textsection{5} of this work.  

\section{\protect\centering $\S 3$ Energy Inequalities}

In this section we derive two inequalities in order to gain enough control over (\ref{rrr})-(\ref{ttt}) to prove the theorem.  That is, from these inequalities we obtain \emph{a priori} estimates that hold for smooth solutions and then prove the existence result using Solonnikov's short-time solution.  The first inequality we show is the classical entropy inequality adapted to the context of a multifluid, while the second is an additional energy inequality derived using a technique discovered by D.~Bresch and B.~Desjardins that effectively fixes the form of the viscosity coefficient $\nu(\rho,\mu)$.   

A simple calculation is required in order to obtain the classical entropy inequality (in the sense of \cite{BD2} and \cite{MV3}).  That is, multiplying the momentum equation (\ref{sss}) by $u$ and integrating we find, \begin{equation}\begin{aligned}\label{energy}\frac{d}{dt}\int_{\mathbb{R}}\Big\{\rho\frac{u^{2}}{2}+\mathscr{E}(\rho,\mu )\Big\}dx+\int_{\mathbb{R}}\nu(\rho,\mu)|\partial_{x}u|^{2}dx\leq 0.\end{aligned} \end{equation} Here $\mathscr{E}(\rho,\mu)$ is the \emph{internal energy} functional effectively tempered by a fixed constant reference density $\tilde{\rho}<\infty$ and a fixed constant reference mass fraction $\tilde{\mu}\leq C$, given by \begin{equation}\label{genenerg}\mathscr{E}(\rho,\mu)=\rho\int_{\tilde{\rho}}^{\rho} \bigg\{\frac{p(s,\mu)-p(\tilde{\rho},\mu)}{s^{2}}\bigg\}ds+p(\tilde{\rho},\tilde{\mu})-p(\tilde{\rho},\mu).\end{equation} 

Let us make this calculation precise.  First we restrict to the first two terms of (\ref{sss}) to notice that, $$\partial_{t}(\rho u)+\partial_{x}(\rho u^{2})=\rho\partial_{t}u+\rho u\partial_{x}u+u\left(\partial_{t}\rho +\partial_{x}(\rho u)\right),$$ where subsequently multiplying through by a factor of $u$ and integrating gives, $$\int_{\mathbb{R}}\bigg\{u^{2}\left(\partial_{t}\rho +\partial_{x}(\rho u)\right)+\frac{1}{2}\Big(\rho\partial_{t}u^{2}+\rho u\partial_{x}u^{2}\Big)\bigg\}dx.$$  This can be easily rewritten using (\ref{rrr}), as $$\begin{aligned}\int_{\mathbb{R}}\left(u\partial_{t}(\rho u)+u\partial_{x}(\rho u^{2})\right)dx=\frac{1}{2}\int_{\mathbb{R}}\Big\{\partial_{t}\left(\rho u^{2}\right)+\partial_{x}\left(\rho u^{3}\right)\Big\}dx.\end{aligned}$$

Likewise the pressure term $p_{x}$ from (\ref{sss}) is  multiplied through by a factor of $u$ and integrated.  In particular, the form this term takes in (\ref{energy}) is derived from a pressure $p(\rho,\mu)$ that satisfies a conservation law (shown in \textsection{4}) for a tempered internal energy $\mathscr{E}(\rho,\mu)$.  To see this, first notice that for any function of $\rho$ and $\mu$ we have, $$\begin{aligned}\label {timetrans}\int_{\mathbb{R}}\partial_{t}\mathscr{E}(\rho,\mu)dx &=\int_{\mathbb{R}}\partial_{\rho}\mathscr{E}(\rho,\mu)\partial_{t}\rho dx+\int_{\mathbb{R}}\partial_{\mu}\mathscr{E}(\rho,\mu)\partial_{t}\mu dx \\ & = -\int_{\mathbb{R}}u \partial_{\mu}\mathscr{E}(\rho,\mu) \partial_{x}\mu dx-\int_{\mathbb{R}}\partial_{\rho}\mathscr{E}(\rho,\mu)\partial_{x}(\rho u)dx \\ &=-\int_{\mathbb{R}}u\partial_{\mu}\mathscr{E}(\rho,\mu)\partial_{x}\mu dx-\int_{\mathbb{R}}\rho\partial_{\rho}\mathscr{E}(\rho,\mu)\partial_{x}u dx \\ & \quad\quad -\int_{\mathbb{R}}u\partial_{\rho}\mathscr{E}(\rho,\mu)\partial_{x}\rho dx.\end{aligned}$$  But here, since $\partial_{x}\mathscr{E}(\rho,\mu)=\partial_{\rho}\mathscr{E}\partial_{x}\rho +\partial_{\mu}\mathscr{E}\partial_{x}\mu$, we can write, $$\begin{aligned}\int_{\mathbb{R}}\partial_{t}\mathscr{E}(\rho,\mu)dx & =-\int_{\mathbb{R}}u\partial_{\mu}\mathscr{E}(\rho,\mu)\partial_{x}\mu dx -\int_{\mathbb{R}}\rho\partial_{\rho}\mathscr{E}(\rho,\mu)\partial_{x}u dx \\ & \quad \quad -\int_{\mathbb{R}}u\partial_{\rho}\mathscr{E}(\rho,\mu)\partial_{x}\rho dx \\ & =-\int_{\mathbb{R}}u\partial_{x} \mathscr{E}(\rho,\mu)dx-\int_{\mathbb{R}}\rho\partial_{\rho}\mathscr{E}(\rho,\mu)\partial_{x}u dx \\ &=\int_{\mathbb{R}}\partial_{x}u\Big\{\mathscr{E}(\rho,\mu)-\rho\partial_{\rho}\mathscr{E}(\rho,\mu)\Big\}dx,\end{aligned}$$ which gives, \begin{equation}\begin{aligned}\label{whynot}\int_{\mathbb{R}}\partial_{t}\mathscr{E}(\rho,\mu)dx & =\int_{\mathbb{R}}u\partial_{x}\Big\{\rho \partial_{\rho}\mathscr{E}(\rho,\mu)-\mathscr{E}(\rho,\mu)\Big\}dx.\end{aligned}\end{equation}  Using (\ref{genenerg}) we find \begin{equation}\label{internalproper}\mathscr{E}(\rho,\mu) = \rho\partial_{\rho}\mathscr{E}(\rho,\mu) + p(\tilde{\rho},\tilde{\mu})-p(\rho,\mu),\quad\mathrm{where}\quad \mathscr{E}(\tilde{\rho},\tilde{\mu})=0,\end{equation} such that computing $\rho\partial_{\rho}\mathscr{E}(\rho,\mu)-\mathscr{E}(\rho,\mu)$ arrives with the desired equality, \begin{equation}\label{temperpres}\frac{d}{dt}\int_{\mathbb{R}}\mathscr{E}(\rho,\mu)dx =\int_{\mathbb{R}}u\partial_{x}p(\rho,\mu)dx.\end{equation}

This internal energy $\mathscr{E}$ over $\mathbb{R}$ arises in \cite{Hoff4} for the single component case, where there $G(\rho,\rho')$ is set as the \emph{potential energy density} and treated in a similar fashion.  Note that as $\mathscr{E}(\rho,\mu)$ is tempered with respect to a reference density $\tilde{\rho}$ and a reference fractional mass $\tilde{\mu}$, this is all that is needed to control the sign on the internal energy $\mathscr{E}(\rho,\mu)$, with the only qualification coming from \textsection{3} which gives cases on the limits of integration.

It is further worth mentioning that the above terms comprise an entropy $\mathscr{S}(\rho, u,\mu)$ of the system (as well as the entropy term of inequality (\ref{energy})), where we write the integrable function, $$\mathscr{S}(\rho, u,\mu) =\frac{m^{2}}{2\rho}+\mathscr{E}(\rho,\mu).$$
 
The final step in recovering (\ref{energy}) is to calculate the remaining diffusion term, which follows directly upon integration by parts.  That is, after multiplying through by $u$ and integrating by parts we see that $$-\int_{\mathbb{R}}u \partial_{x} \big(\nu(\rho,\mu)\partial_{x}u\big)dx=\int_{\mathbb{R}}\nu(\rho,\mu)(\partial_{x}u)^{2}dx$$ which leads to the result; namely (\ref{energy}).

\subsection{3.1 Additional Energy Inequality}

The following lemma provides the second energy inequality that we use in order to prove the theorem.

\newtheorem*{l1vc}{Lemma 3.1}
\begin{l1vc}
For solutions of (\ref{rrr})-(\ref{ttt}) we have \begin{equation}\label{lemma1}\frac{d}{dt}\int_{\mathbb{R}}\Big\{\frac{\rho}{2}\big|u + \rho^{-1}\partial_{x}\psi(p)\big|^{2} + \mathscr{E}(\rho,\mu)\Big\}dx+\int_{\mathbb{R}}\rho^{-1}\psi^{\prime}(p)\big(\partial_{x} p(\rho,\mu)\big)^{2} dx = 0, \end{equation} providing the following constraint on the viscosity $\nu(\rho,\mu)$: \begin{equation} \label{dummy}\nu(\rho,\mu)=\rho\partial_{\rho}p\psi^{\prime}(p).\end{equation}
\end{l1vc}

\begin{proof} Take the continuity equation and the transport equation in $\mu$ and multiply through by derivatives of a function of the pressure $\psi(p)$ such that, $$\begin{aligned}&\partial_{\rho}\psi(p)\Big\{\partial_{t}\rho+\partial_{x}(\rho u)\Big\}=0,\\ & \ \partial_{\mu}\psi(p)\Big\{\partial_{t}\mu +u\partial_{x}\mu\Big\}=0,\end{aligned}$$ where adding the components together gives, $$\partial_{t}\psi(p)+u\partial_{x}\psi(p)+\rho\partial_{\rho}\psi (p)\partial_{x}u=0.$$  A derivation in $x$ provides that $$\partial_{t}\big(\partial_{x}\psi (p)\big)+\partial_{x}\big(u\partial_{x}\psi(p)\big)+\partial_{x}\big(\rho\partial_{\rho}\psi(p)\partial_{x}u\big)=0,$$ which we expand to $$\partial_{t}\big(\rho\rho^{-1}\partial_{x}\psi (p)\big)+\partial_{x}\big(\rho\rho^{-1}u\partial_{x}\psi (p)\big)+\partial_{x}\big(\rho\partial_{\rho}\psi(p)\partial_{x}u\big)=0,$$ such that adding it back to the momentum equation (\ref{sss}) and applying condition (\ref{dummy}) arrives with $$\partial_{t}\big(\rho\big\{u+\rho^{-1}\partial_{x}\psi (p)\big\}\big)+\partial_{x}\big(\rho u\big\{u+\rho^{-1}\partial_{x}\psi(p)\big\}\big)+\partial_{x}p(\rho,\mu)=0. $$  Multiplying this by $\big(u+\rho^{-1}\partial_{x}\psi(p)\big)$ then gives, $$\begin{aligned} \frac{1}{2}\partial_{t} \big\{\rho \big|u+\rho^{-1}\partial_{x}\psi (p)\big|^{2}\big\}+\frac{1}{2}\partial_{x}\big\{\rho u & \big|u+\rho^{-1}\partial_{x}\psi(p)\big|^{2}\big\} \\ & +\big\{u+\rho^{-1}\partial_{x}\psi (p)\big\}\partial_{x} p(\rho,\mu) = 0, \end{aligned}$$ which when integrated becomes $$\frac{d}{dt}\int_{\mathbb{R}}\Big\{\frac{\rho}{2}\big|u+\rho^{-1}\partial_{x}\psi(p)\big|^{2}+\mathscr{E}(\rho,\mu)\Big\}dx+\int_{\mathbb{R}}\rho^{-1}\psi^{\prime}(p)\big(\partial_{x}p(\rho,\mu)\big)^{2}dx=0,$$ completing the proof.\end{proof}

\section{\protect\centering $\S 4$ Establishing the Existence Theorem}
In this section our aim is to apply the inequalities in \textsection{3} predicated on the formulation in \textsection{2} to acquire the existence half of the theorem.  However, in order to do this we must first confirm that the energy inequalities satisfy the appropriate bounds.  Let us demonstrate this principle for both (\ref{energy}) and (\ref{lemma1}) in the form of the following lemma.  

\newtheorem*{l2vc}{Lemma 4.1}
\begin{l2vc}
For any solution $(\rho,u,\mu)$ of (\ref{rrr})-(\ref{ttt}) verifying, \begin{equation}\label{initcent}\int_{\mathbb{R}}\Big\{\rho_{0}\frac{u_{0}^{2}}{2}+\mathscr{E}(\rho_{0},\mu_{0})\Big\}dx<+\infty\end{equation} $$\mathit{and}$$ \begin{equation}\label{initclem}\int_{\mathbb{R}}\Big\{\frac{\rho_{0}}{2}\Big|u_{0}+\frac{\partial_{x}\psi(p_{0})}{\rho_{0}}\Big|^{2}+\mathscr{E}(\rho_{0},\mu_{0})\Big\}dx < +\infty,\end{equation} we have that \begin{equation}\label{12vc}\esssup_{[0,T]}\int_{\mathbb{R}}\Big\{\rho\frac{u^{2}}{2}+\mathscr{E}(\rho,\mu)\Big\}dx+\int_{0}^{T}\int_{\mathbb{R}}\nu(\rho,\mu)|\partial_{x}u|^{2}dxdt\leq C,\end{equation} $$\mathit{and}$$ \begin{equation}\begin{aligned}\label{12vc2}\esssup_{[0,T]}\int_{\mathbb{R}}\Big\{\frac{\rho}{2}\Big|u+\frac{\partial_{x}\psi(p)}{\rho}\Big|^{2}+\mathscr{E}(\rho,\mu)\Big\}dx+\int_{0}^{T}\int_{\mathbb{R}}\frac{\psi^{\prime}(p)}{\rho}|&\partial_{x}p|^{2}dxdt \leq C.\end{aligned}\end{equation} 
\end{l2vc}

\begin{proof} It suffices if every term on the left side of both inequality (\ref{12vc}) and (\ref{12vc2}) can be shown to be nonnegative.
    
First notice that we clearly have that $\rho u^{2}\geq 0$ for any barotropic fluid over $\mathbb{R}$, since $\rho$ is strictly nonnegative.  To check that $\mathscr{E}(\rho,\mu)\geq 0$ we simply refer to the definition given in (\ref{internalproper}).  Indeed $\left(\frac{p(s,\mu)-p(\tilde{\rho},\mu)}{s^{2}}\right)\geq 0$ when $\rho\geq \tilde{\rho}$ and $\left(\frac{p(s,\mu)-p(\tilde{\rho},\mu)}{s^{2}}\right)\leq 0$ for $\rho\leq \tilde{\rho}$, which implies $$\int_{\tilde{\rho}}^{\rho}\frac{p(s,\mu)-p(\tilde{\rho},\mu)}{s^{2}}ds\geq 0.$$  Together with (\ref{genenerg}) this gives that $\mathscr{E}(\rho,\mu)\geq 0$ . 

Next we check the viscosity coefficient $\nu(\rho,\mu)$.  Here the positivity follows from (\ref{visc}), where again the pressure is increasing in $\rho$ satisfying (\ref{rhoincrease}) and the density is positive definite away from the vacuum solution (which we show is forbidden due to proposition 4.1), so for a $\psi'(p)$ satisfying (\ref{rho2}) we see that $\psi'(p)\geq 0$.  Similarly, the last term on the right in (\ref{12vc2}) follows away from vacuum, where again we only rely upon the fact from \textsection{2} that $\psi'(p)\geq 0$. \end{proof}

These results provide the estimates that we use for the remainder of the paper.  That is, it is well-known (for example see Theorem 7.2 in \cite{PLL2} and the results in \cite{MV1}) that the existence of a global strong solution to the system (\ref{rrr})-(\ref{sss}) follows by regularity analysis in tandem with (\ref{energy}) and (\ref{lemma1}).  Below we present a similar approach for the case of a mixing multicomponent fluid (\ref{rrr})-(\ref{ttt}) using only what we have found above; namely, that (\ref{energy}) and (\ref{lemma1}) provide the following \emph{a priori} bounds: \begin{equation}\begin{aligned}\label{apri1}\|\sqrt{\nu(\rho,\mu)}&\partial_{x}u\|_{L^{2}(0,T;L^{2}(\mathbb{R}))}\leq C,\\\|\sqrt{\rho}u\|&_{L^{\infty}(0,T;L^{2}(\mathbb{R}))}\leq C,\\\|\mathscr{E}(\rho,\mu)&\|_{L^{\infty}(0,T;L^{1}(\mathbb{R}))}\leq C,\end{aligned}\end{equation} along with, \begin{equation}\begin{aligned}\label{apri2}\|(\partial_{x}\psi(p)/\sqrt{\rho})\|&_{L^{\infty}(0,T;L^{2}(\mathbb{R}))}\leq C,\\\|(\psi'(p)/\rho)^{1/2}\partial_{x}p(\rho,&\mu)\|_{L^{2}(0,T;L^{2}(\mathbb{R}))}\leq C.\end{aligned}\end{equation}  We will use these inequalitites extensively for the remainder of the paper.

As a remark, if we denote the \emph{internal energy density} $e$ as being characterized by the relations, \begin{equation}\begin{aligned}\label{energydensity}& \rho e =\rho\int_{\tilde{\rho}}^{\rho}\frac{\partial e(s,\mu)}{\partial\rho} ds \quad\mathrm{with}\quad e(\tilde{\rho},\mu)=0,\\ &\mathrm{and}\quad \partial_{\rho}e(\rho,\mu)=\rho^{-2}p(\rho,\mu)-\rho^{-2} p(\tilde{\rho},\mu),\end{aligned}\end{equation} then $e$ is closely related to the \emph{specific internal energy} $e_{s}$, defined by $$e_{s}(\rho)\equiv\int_{1}^{\rho}\frac{p(s)}{s^{2}}ds,$$ which is provided for the single barotropic compressible fluid case in \cite{NS} and \cite{F1}; but in the multifluid context, since the internal energy $\mathscr{E}$ is tempered up to some constant reference density $\tilde{\rho}$, the usual form of the specific internal energy inherits a tempering in $\tilde{\rho}$ as well, which is what is provided here by the function $e$.  We also note that the tempered internal energy $\mathscr{E}$ now satisfies the following conservation form as mentioned in \textsection{3}, \begin{equation}\label{where}\partial_{t}\mathscr{E}(\rho,\mu)+\partial_{x}(\mathscr{E}(\rho,\mu) u) + \Big\{\rho^{2}\partial_{\rho}e(\rho,\mu)+p(\tilde{\rho},\mu)-p(\tilde{\rho},\tilde{\mu})\Big\}\partial_{x}u = 0,\end{equation} where it is easy to confirm that upon integration this recovers (\ref{temperpres}).   

\subsection{4.1 Bounds on the Density}

For the existence theorem we need to establish a bound for the density in the space $L^{\infty}(0,T;\dot{H}^{1}(\mathbb{R}))$.  To achieve this we first establish uniform bounds on the density.

\newtheorem*{rhoprop}{Proposition 4.1}
\begin{rhoprop}  For every $T>0$ there exist two distinct positive constants $\underline{\varrho}$ and $\overline{\varrho}$ such that \begin{equation}\label{lowup}\underline{\varrho}\leq \rho(t,x)\leq\overline{\varrho}\quad\quad\forall (t,x)\in [0,T]\times\mathbb{R}.\end{equation}
\end{rhoprop}

Showing this proposition requires the following three lemmas which provide the groundwork for its subsequent proof.

\newtheorem*{rho1g}{Lemma 4.2}
\begin{rho1g} Let $\mathscr{F}\geq 0$ be a function defined on $[0,+\infty)\times\mathbb{R}$ where $\mathscr{F}(\cdot,x)$ is uniformly continuous with respect to $x\in\mathbb{R}$ and where there exists a $\delta>0$ with $\mathscr{F}(0,x)>\delta$ for any $x$.  Then there exists an $\epsilon >0$ such that for any constant $\bar{C} >0$ there exists a constant $K >0$ so that for any nonnegative function $f$ verifying $\int_{\mathbb{R}}\mathscr{F}(f(x),x) dx \leq \bar{C}$, and for any $x_{0}\in\mathbb{R}$, there exists a point $x_{1}\in I=[x_{0}-K,x_{0}+K]$ such that $f(x_{1})>\epsilon$.     
\end{rho1g}

\begin{proof}   For any fixed $\mathscr{F}$ there exists a $\tilde{C}$ such that for all $y \leq \epsilon$ we have, $$\mathscr{F}(y,x) \geq \frac{1}{2\tilde{C}}$$ since $\mathscr{F}(0,x)>\delta$ for any $x$ and $\mathscr{F}$ is uniformly continuous in $x$.  Let us fix $\bar{C}>0$ and define \begin{equation}\label{K}K=2\bar{C}\tilde{C}.\end{equation}  We show that this $K$ verifies the desired properties.   Here we utilize a proof by contradiction in the spirit of \cite{MV3}.  Assume that we can find a nonnegative function $f$ verifying $\int_{\mathbb{R}}\mathscr{F}(f(x),x)dx\leq\bar{C}$ and an $x_{0}\in\mathbb{R}$ with $$\esssup_{x\in I} f \leq \epsilon,$$ where $I=[x_{0}-K,x_{0}+K]$.  Since $\mathscr{F}\geq0$ this implies  $$\bar{C}\geq\int_{\mathbb{R}}\mathscr{F}(f(x),x)dx \geq \int_{I}\mathscr{F}(f(x),x)dx\geq\int_{I}\frac{1}{2\tilde{C}}dx,$$ which yields $\bar{C}\geq K/\tilde{C}$ in contradiction to (\ref{K}).\end{proof}    

Additionally we require the following technical lemma.

\newtheorem*{tech}{Lemma 4.3}
\begin{tech} Providing (\ref{rho2}) then (\ref{genenerg}) yields, $$\begin{aligned}&\rho^{\hat{\gamma}}+ C\rho\leq C+1 \quad\mathit{for}\quad\rho\leq 1,\\ \rho^{\check{\gamma}}+&\frac{\rho}{C}\leq C+ C\mathscr{E}(\rho,\mu)\quad\mathit{for}\quad\rho\geq 1.\end{aligned}$$     
\end{tech}

\begin{proof} Trivially, when $\rho\leq 1$ we have that $\rho^{\hat{\gamma}}+ C\rho\leq C+1$.  When $\rho\geq 1$ we use (\ref{rho2}) to expand $\mathscr{E}(\rho,\mu)$ where (\ref{genenerg}) gives that as $\rho\rightarrow\infty$ the $\rho^{\check{\gamma}}$ dominates such that scaling the constant correctly provides the result.\end{proof}

Now we are able to find uniform positive bounds $\underline{\varrho}$ and $\overline{\varrho}$ on the density which inherently preclude the vacuum and concentration states.

\newtheorem*{rho2}{Lemma 4.4}
\begin{rho2} Assume that (\ref{gamma1}), (\ref{gamma2}) and (\ref{rho2})-(\ref{rhomub}) are satisfied and let \begin{equation}\label{xidef}\partial_{x}\xi(\rho)=\mathbbm{1}_{\{\rho\leq1\}}\partial_{x}\rho^{-\eta}+\mathbbm{1}_{\{\rho\geq 1\}}\partial_{x}\rho^{\sigma}.\end{equation}  Then there exists an $\eta>0$ and $\sigma>0$ such that for any $K>0$ there exists a $C_{K}$ with \begin{equation}\label{44lem}\|\partial_{x}\xi(\rho)\|_{L^{\infty}(0,T;L^{1}(I)}\leq C_{K}\end{equation} for every  $x_{0}\in\mathbb{R}$ and $I=[x_{0}-K,x_{0}+K]$.\end{rho2}

\begin{proof} First recall that the pressure satisfies $$\partial_{x}p(\rho,\mu)= \partial_{\rho}p(\rho,\mu)\partial_{x}\rho+\partial_{\mu}p(\rho,\mu)\partial_{x}\mu.$$  Here we are concerned with two cases, namely when $\rho\leq 1$ and when $\rho\geq 1$.  For the case when $\rho\leq 1$ we multiply through by $\rho^{-1/2}p(\rho,\mu)^{-\underline{\alpha}}$ where $\underline{\alpha}$ is given by (\ref{gamma3}), which yields \begin{equation}\label{xi1} \frac{\partial_{\rho}p(\rho,\mu)\partial_{x}\rho}{\sqrt{\rho}p(\rho,\mu)^{\underline{\alpha}}}=\frac{\partial_{x}p(\rho,\mu)}{\sqrt{\rho}p(\rho,\mu)^{\underline{\alpha}}}-\frac{\partial_{\mu}p(\rho,\mu)\partial_{x}\mu}{\sqrt{\rho}p(\rho,\mu)^{\underline{\alpha}}}.\end{equation}  Likewise for $\rho\geq 1$ we multiply through by $p(\rho,\mu)^{-\overline{\alpha}}\rho^{-1/2}$ given $\overline{\alpha}$ from (\ref{gamma4}) such that \begin{equation}\label{xi1b}\frac{\partial_{\rho}p(\rho,\mu)\partial_{x}\rho}{\sqrt{\rho}p(\rho,\mu)^{\overline{\alpha}}}=\frac{\partial_{x}p(\rho,\mu)}{\sqrt{\rho}p(\rho,\mu)^{\overline{\alpha}}}-\frac{\partial_{\mu}p(\rho,\mu)\partial_{x}\mu}{\sqrt{\rho}p(\rho,\mu)^{\overline{\alpha}}}.\end{equation}  In order to demonstrate the lemma we will control the right hand side of both (\ref{xi1}) and (\ref{xi1b}) such that each is bounded in $L^{\infty}(0,T;L^{1}_{loc}(\mathbb{R}))$ for any point in $\mathbb{R}$ up to some fixed subinterval $I$.

Towards this, we first show that $\rho^{-1}\partial_{x}\mu$ is bounded in $L^{\infty}(0,T;L^{\infty}(\mathbb{R}))$.  That is, take a derivation in $x$ of (\ref{ttt}) in order to write \begin{equation}\label{varphi1}\partial_{t}(\rho\rho^{-1}\partial_{x}\mu)+\partial_{x}(\rho u\rho^{-1}\partial_{x}\mu)=0,\end{equation} such that multiplying through by a function $\vartheta'(\rho^{-1}\partial_{x}\mu)=\vartheta'$ achieves $$\partial_{t}(\rho\vartheta(\rho^{-1}\partial_{x}\mu))+\partial_{x}(\rho u \vartheta(\rho^{-1}\partial_{x}\mu))=0.$$  Now choose $\vartheta'$ such that for every test function $\vartheta(\rho^{-1}\partial_{x}\mu)\in \mathcal{D}(\mathbb{R})$ with compact support, the function $\vartheta$ vanishes almost everywhere over the finite interval $\mathcal{I}=[-M,M]$, with $M$ a constant.  Upon integration this implies $$\int_{0}^{T}\frac{d}{dt}\int_{\mathbb{R}}\rho\vartheta(\rho^{-1}\partial_{x}\mu)dxdt = 0,$$ such that for an appropriate choice of initial condition, where $\rho_{0}^{-1}\partial_{x}\mu_{0}\in \mathcal{I}$, we find $$\esssup_{[0,T]}\int_{\mathbb{R}}\rho\vartheta(\rho^{-1}\partial_{x}\mu)dx = 0.$$  This implies that $\rho\vartheta(\rho^{-1}\partial_{x}\mu)=0$ almost everywhere for all $(t,x)\in (0,T)\times\mathbb{R}$, and so we can conclude that the argument of $\vartheta$ takes values over the interval $\mathcal{I}$, or more clearly that for $\rho$ a.e. $|\rho^{-1}\partial_{x}\mu|\leq M$.  This is then enough to educe the norm: $$\|\rho^{-1}\partial_{x}\mu\|_{L^{\infty}(0,T;L^{\infty}(\mathbb{R}))}\leq M.$$ 
  
However, this is not yet enough to control the last term on the right for the two cases.  In (\ref{xi1}) applying (\ref{rho2}) and (\ref{rhomub}) further provides $$\frac{\sqrt{\rho}\partial_{\mu}p(\rho,\mu)}{p(\rho,\mu)^{\underline{\alpha}}}\leq C_{0}\rho^{\check{\gamma}- \underline{\alpha}\hat{\gamma}+\frac{1}{2}}\quad\quad\mathrm{for}\quad\quad\rho\leq 1$$ for a positive constant $C_{0}$.  Using (\ref{gamma3}) from above we have that $\check{\gamma}\geq\underline{\alpha}\hat{\gamma}-1/2$, and so the positivity of the exponent gives $$C_{0}\rho^{\check{\gamma}-\underline{\alpha}\hat{\gamma}+\frac{1}{2}}\leq C\quad\quad\mathrm{for}\quad\quad\rho\leq 1,$$ which leads to, \begin{equation}\bigg\|\mathbbm{1}_{\{\rho\leq 1\}}\frac{\sqrt{\rho}\partial_{\mu}p(\rho,\mu)}{p(\rho,\mu)^{\underline{\alpha}}}\bigg\|_{L^{\infty}(0,T;L^{\infty}(\mathbb{R}))}\leq C.\end{equation}

Similarly for (\ref{xi1b}) we apply (\ref{rho2}) and (\ref{rhomub}) to see that, $$\frac{\sqrt{\rho}\partial_{\mu}p(\rho,\mu)}{p(\rho,\mu)^{\overline{\alpha}}}\leq C_{0}\rho^{\hat{\gamma}-\overline{\alpha}\check{\gamma}+\frac{1}{2}}\quad\quad\mathrm{for}\quad\quad \rho\geq 1.$$  Notice that since (\ref{gamma4}) provides $\check{\gamma}\geq\hat{\gamma}-\overline{\alpha}\check{\gamma}+1/2$, then applying lemma 4.3 implies $$\mathbbm{1}_{\{\rho\geq 1\}}\left(\frac{\sqrt{\rho}\partial_{\mu}p(\rho,\mu)}{p(\rho,\mu)^{\overline{\alpha}}}\right)\leq C+C\mathscr{E}(\rho,\mu).$$  Integrating over $I$ gives $$\int_{I} \bigg|\mathbbm{1}_{\{\rho\geq 1\}} \frac{\sqrt{\rho}\partial_{\mu}p(\rho,\mu)}{p(\rho,\mu)^{\overline{\alpha}}}\bigg|dx\leq 2KC+C\int_{\mathbb{R}}\mathscr{E}(\rho,\mu)dx,$$ such that applying (\ref{apri1}) establishes $$\bigg\|\mathbbm{1}_{\{\rho\geq 1\}}\frac{\sqrt{\rho}\partial_{\mu}p(\rho,\mu)}{p(\rho,\mu)^{\overline{\alpha}}}\bigg\|_{L^{\infty}(0,T;L_{loc}^{1}(\mathbb{R}))}\leq C_{K},$$ for $C_{K}$ a constant depending only on $K$.

Now consider the $\partial_{x}p$ term in (\ref{xi1}) where here again we treat the two cases $\rho\leq 1$ and $\rho\geq 1$ separately.  For the case $\rho\leq 1$ notice that we have by the bound on $\psi'(p)$ in (\ref{rho2}) that $$ \bigg|\mathbbm{1}_{\{\rho\leq 1\}}\frac{\partial_{x}p(\rho,\mu)}{\sqrt{\rho}p(\rho,\mu)^{\underline{\alpha}}}\bigg|=C|\mathbbm{1}_{\{\rho\leq 1\}}\rho^{-1/2}\partial_{x}p(\rho,\mu)^{1-\underline{\alpha}}|\leq C|\mathbbm{1}_{\{\rho\leq 1\}}\rho^{-1/2}\partial_{x}\psi(p)|.$$  Upon integration (\ref{apri2}) gives $$\begin{aligned}\int_{\mathbb{R}}|\mathbbm{1}_{\{\rho\leq 1\}}\rho^{-1/2}\partial_{x}p(\rho,\mu)^{1-\underline{\alpha}}|^{2}dx &\leq C \int_{\mathbb{R}}|\mathbbm{1}_{\{\rho\leq 1\}}\rho^{-1/2}\partial_{x}\psi(p)|^{2}dx \\ &\leq C,\end{aligned}$$ and so we obtain $$\|\mathbbm{1}_{\{\rho\leq 1\}}\rho^{-1/2}p(\rho,\mu)^{-\underline{\alpha}}\partial_{x}p(\rho,\mu)\|_{L^{\infty}(0,T;L^{2}(\mathbb{R}))}\leq C.$$ 
 
Similarly for $\rho\geq 1$ we apply (\ref{rho2}), giving $$\bigg|\mathbbm{1}_{\{\rho\geq 1\}}\frac{\partial_{x}p(\rho,\mu)}{\sqrt{\rho}p(\rho,\mu)^{\overline{\alpha}}}\bigg|=C|\mathbbm{1}_{\{\rho\geq 1\}}\rho^{-1/2}\partial_{x}p(\rho,\mu)^{1-\overline{\alpha}}|\leq C|\mathbbm{1}_{\{\rho\geq 1\}}\rho^{-1/2}\partial_{x}\psi(p)|,$$ such that integrating and utilizing (\ref{apri2}) yields $$\begin{aligned}\int_{\mathbb{R}}\bigg|\mathbbm{1}_{\{\rho\geq 1\}}\frac{\partial_{x}p(\rho,\mu)}{\sqrt{\rho}p(\rho,\mu)^{\overline{\alpha}}}\bigg|dx &\leq C\int_{\mathbb{R}}|\mathbbm{1}_{\{\rho\geq 1\}}\rho^{-1/2}\partial_{x}\psi(p)|dx\\ &\leq C,\end{aligned}$$ and so $$\|\mathbbm{1}_{\{\rho\geq 1\}}\rho^{-1/2}p(\rho,\mu)^{-\overline{\alpha}}\partial_{x}p(\rho,\mu)\|_{L^{\infty}(0,T;L^{1}(\mathbb{R}))}\leq C.$$ 

Combining these results we have thus acquired the important bound on the left sides of (\ref{xi1}) and (\ref{xi1b}): \begin{equation}\begin{aligned}\label{rhopbound}\|\mathbbm{1}_{\{\rho\leq 1\}}&\rho^{-1/2}p(\rho,\mu)^{-\underline{\alpha}}\partial_{\rho}p(\rho,\mu)\partial_{x}\rho\|_{L^{\infty}(0,T;L^{2}_{loc}(\mathbb{R}))} \\ & + \|\mathbbm{1}_{\{\rho\geq 1\}}\rho^{-1/2}p(\rho,\mu)^{-\overline{\alpha}}\partial_{\rho}p(\rho,\mu)\partial_{x}\rho\|_{L^{\infty}(0,T;L^{1}(I))}\leq C_{K}.\end{aligned}\end{equation}

It remains to show that for all $\rho$ we have bounds on some power of the spatial derivative $\rho_{x}$.  First notice that when $\rho\leq 1$ applying (\ref{rho2}) and (\ref{rhopb}) to the left of (\ref{xi1}) provides $$\bigg|\mathbbm{1}_{\{\rho\leq 1\}}\frac{\partial_{\rho}p(\rho,\mu)\partial_{x}\rho}{\sqrt{\rho}p(\rho,\mu)^{\underline{\alpha}}}\bigg|\geq C|\mathbbm{1}_{\{\rho\leq 1\}}\rho^{\hat{\gamma}-\underline{\alpha}\check{\gamma}-3/2}\partial_{x}\rho|= C|\mathbbm{1}_{\{\rho\leq 1\}}\rho^{\hat{\gamma}-\underline{\alpha}\check{\gamma}-3/2}||\partial_{x}\rho|,$$ such that upon squaring and integrating we find $$\begin{aligned}C\int_{\mathbb{R}}|\mathbbm{1}_{\{\rho\leq 1\}}\rho^{\hat{\gamma}-\underline{\alpha}\check{\gamma}-3/2}\partial_{x}\rho|^{2}dx & \leq\int_{\mathbb{R}}\bigg|\mathbbm{1}_{\{\rho\leq 1\}}\frac{\partial_{\rho}p(\rho,\mu)\partial_{x}\rho}{\sqrt{\rho}p(\rho,\mu)^{\underline{\alpha}}}\bigg|^{2}dx\\ &\leq C.\end{aligned}$$  This provides what we desire by way of the following equality: \begin{equation}\begin{aligned}\label{xi2}\|\mathbbm{1}_{\{\rho\leq 1\}}\rho^{\hat{\gamma}-\underline{\alpha}\check{\gamma} -3/2} \partial_{x}\rho\|_{L^{\infty}(0,T;L_{loc}^{2}(\mathbb{R}))}& =C\|\mathbbm{1}_{\{\rho\leq 1\}}\partial_{x}\rho^{\hat{\gamma}-\underline{\alpha}\check{\gamma}-1/2}\|_{L_{loc}^{\infty}(0,T;L^{2}(\mathbb{R}))}\\ &\leq C.\end{aligned}\end{equation}  Thus when applying the condition from (\ref{gamma1}) it follows that $\eta$ satisfies \begin{equation}\label{eta}\eta=\underline{\alpha}\check{\gamma} -\hat{\gamma} + \frac{1}{2}.\end{equation} 

Likewise when $\rho\geq 1$ applying (\ref{rho2}) and (\ref{rhopb}) provides $$\bigg|\mathbbm{1}_{\{\rho\geq 1\}}\frac{\partial_{\rho}p(\rho,\mu)\partial_{x}\rho}{\sqrt{\rho}p(\rho,\mu)^{\overline{\alpha}}}\bigg|\geq C|\mathbbm{1}_{\{\rho\geq 1\}}\rho^{\check{\gamma}-\overline{\alpha}\hat{\gamma}-3/2}\partial_{x}\rho|= C|\mathbbm{1}_{\{\rho\geq 1\}}\rho^{\check{\gamma}-\overline{\alpha}\hat{\gamma}-3/2}||\partial_{x}\rho|,$$ such that integrating over $I$ gives by (\ref{rhopbound}) that $$\begin{aligned}C\int_{I}|\mathbbm{1}_{\{\rho\geq 1\}}\rho^{\check{\gamma}-\overline{\alpha}\hat{\gamma}-3/2}\partial_{x}\rho|dx & \leq  \int_{I} \bigg|\mathbbm{1}_{\{\rho\geq 1\}}\frac{\partial_{\rho}p(\rho,\mu)\partial_{x}\rho}{\sqrt{\rho}p(\rho,\mu)^{\overline{\alpha}}}\bigg|dx \\ &\leq C_{K}.\end{aligned}$$  Here this yields \begin{equation}\label{upper2}\|\mathbbm{1}_{\{\rho\geq 1\}}\rho^{\check{\gamma}-\overline{\alpha}\hat{\gamma}-\frac{3}{2}}\partial_{x}\rho\|_{L^{\infty}(0,T;L^{1}(I))}=C\|\mathbbm{1}_{\{\rho\geq 1\}}\partial_{x}\rho^{\check{\gamma}-\overline{\alpha}\hat{\gamma}-\frac{1}{2}}\|_{L^{\infty}(0,T;L^{1}(I))}\leq C_{K}.\end{equation}  Thus using the condition from (\ref{gamma2}) establishes \begin{equation}\label{sigma}\sigma=\check{\gamma}-\overline{\alpha}\hat{\gamma}-\frac{1}{2}.\end{equation} 

In order to complete the proof all that remains is to add (\ref{xi2}) and (\ref{upper2}) together and apply Minkowski's inequality, which gives $$\|\partial_{x}\xi(\rho)\|_{L^{\infty}(0,T;L^{1}(I))}\leq C_{K}.$$ \end{proof}

We are now able to show Proposition 4.1 by applying the preceding results.

\begin{proof}[Proof of Proposition 4.1] For $t$ fixed set $$\mathscr{F}(y,x)=\begin{cases}\mathscr{E}(y,\mu(t,x))& \quad\mathrm{for} \quad y \leq 1 \\ \mathscr{E}(1,\mu(t,x))&\quad \mathrm{for}\quad y\geq 1\end{cases}$$ such that $y=\rho$.  Next (\ref{genenerg}) together with (\ref{rho2}) shows that $\mathscr{F}(y,x)$ is continuous in $\rho$ uniformly with respect to $x$, and (\ref{apri1}) assures that $$\int_{\mathbb{R}}\mathscr{F}(\rho(t,x),x)dx\leq C.$$  Then the hypothesis of lemma 4.2 is satisfied as long as there exists a $\delta>0$ such that $\mathscr{F}(0,x)>\delta$.  But for $\rho\leq 1$ we can apply (\ref{rho2}) to the form of the internal energy (\ref{genenerg}) to see that as $\rho\rightarrow 0$ we have $\mathscr{E}(\rho,\mu)\geq C$.  Likewise when $\rho=1$ we see that $\mathscr{E}(1,\mu)\geq C_{1}$ for $C_{1}$ a constant.  So we have for a positive $\delta < \inf \{C,C_{1}\}$ that the hypothesis of lemma 4.2 is satisfied.  Then for any $x\in\mathbb{R}$ with $x_{0}=x$ from lemma 4.2 there exists an $x_{1}\in I=[x-K,x+K]$ such that $\rho(t,x_{1})>\epsilon$.  Note that $K$ does not depend on $t$ since $\int_{\mathbb{R}}\mathscr{F}(\rho(t,x),x)dx$ does not depend on time thanks to (\ref{apri1}).  Then the fundamental theorem provides: $$\begin{aligned}|\mathbbm{1}_{\{\rho\leq 1\}}\rho^{-\eta}(x)| & \leq |\epsilon^{-\eta}| + \int_{I}|\mathbbm{1}_{\{\rho\leq 1\}}\partial_{x}\rho^{-\eta}|dx.\end{aligned}$$  Since $K$ does not depend on time, lemma 4.4 gives that the right hand side is bounded uniformly in $x$ and $t$.

For the upper bound, again fix $t$ and now set $$\mathscr{F}(y,x)=\mathscr{E}\left(\frac{1+\tilde{\rho}}{1+y},\mu(t,x)\right)\quad\forall y\geq 0 $$ such that $y=1/\rho$.  Again (\ref{genenerg}) and (\ref{rho2}) provide that $\mathscr{F}(y,x)$ is continuous in $\rho$ uniformly with respect to $x$.  Additionally we find that both $\mathscr{F}(1/\tilde{\rho},x)=\mathscr{E}(\tilde{\rho},\mu)\geq C$ and that $\mathscr{F}(0,x)> C_{1}$ by applying (\ref{rho2}) to (\ref{genenerg}), which provides an admissible $\delta$.  Now, upon defining a function $\varpi = \rho (1+\tilde{\rho})/(\rho+1)$, then there exists a constant $C>0$ such that $$\mathscr{E}(\varpi,\mu)\leq C\mathscr{E}(\rho,\mu),$$ which can be shown using (\ref{genenerg}) and checking the formula for $|\rho-\tilde{\rho}|\leq \frac{\tilde{\rho}}{2}$, $\rho\leq \frac{\tilde{\rho}}{2}$ and $\rho\geq \frac{3}{2}\tilde{\rho}$ thanks to (\ref{rhopb}).  Then (\ref{apri1}) is enough to deduce that $$\int_{\mathbb{R}}\mathscr{F}(\rho(t,x)^{-1},x)dx\leq C.$$  Hence for any $x\in\mathbb{R}$ we can use lemma 4.2 setting $x_{0}=x$ such that there exists an $x_{1}\in [x-K,x+K]$ with $\rho(t,x_{1})\leq\epsilon^{-1}$.  Again notice that $K$ does not depend on $t$ since (\ref{apri1}) is uniform in time.  Then by the fundamental theorem and lemma 4.4 we obtain $$\begin{aligned}|\mathbbm{1}_{\{\rho\geq 1\}}\rho^{\sigma}(x)| & \leq |\epsilon^{-\sigma}| + \int_{I}|\mathbbm{1}_{\{\rho\geq 1\}}\partial_{x}\rho^{\sigma}|dx. \end{aligned}$$  Again since $K$ does not depend on time, lemma 4.4 gives the right side bounded uniformly in $x$ and $t$ which completes the proof of proposition 4.1. \end{proof}

We proceed by showing the important corollary to this proposition.

\newtheorem*{rho4}{Corollary}
\begin{rho4} Assume that (\ref{gamma1})-(\ref{rho2}) and (\ref{rhopb})-(\ref{rhomub}) are satisfied, then $$\rho\in L^{\infty}(0,T;\dot{H}^{1}(\mathbb{R})).$$ 
\end{rho4}

\begin{proof} Lemma 4.4 provides the appropriate framework.  Thus we will show the bound separately for the cases $\rho\leq 1$ and $\rho\geq 1$.

For $\rho\leq 1$ applying (\ref{rho2}), (\ref{rhopb}) and (\ref{rhomub}) we calculate \begin{equation}\begin{aligned}\label{bal6}\partial_{x}\rho^{-\eta} & =\rho^{-\eta -1}\partial_{x}\rho \\ &=\rho^{-\eta -1} \left(\frac{\partial_{x}p(\rho,\mu)-\partial_{\mu}p(\rho,\mu)\partial_{x}\mu}{\partial_{\rho}p(\rho,\mu)}\right)\\ &\leq \rho^{-\eta-1}\left(\frac{\partial_{x}p(\rho,\mu)}{C\rho^{\hat{\gamma}-1}}\right)-C\rho^{-\eta-\hat{\gamma}}\partial_{\mu}p(\rho,\mu)\partial_{x}\mu \\ &\leq C\rho^{-\underline{\alpha}\check{\gamma} -1/2}\partial_{x}p(\rho,\mu)-C\rho^{-\eta-\hat{\gamma}+\check{\gamma}+1}\left(\frac{\partial_{x}\mu}{\rho}\right).\end{aligned}\end{equation}  Squaring both sides gives $$\begin{aligned}(\partial_{x}\rho)^{2} &\leq\rho^{2+2\eta}\left( C\rho^{-\underline{\alpha}\check{\gamma} -1/2}\partial_{x}p(\rho,\mu)-C\rho^{-\eta-\hat{\gamma}+\check{\gamma}+1}\left(\frac{\partial_{x}\mu}{\rho}\right)\right)^{2}.\end{aligned}$$  Integrating, applying (\ref{rho2}) and utilizing H\"{o}lder's inequality yields, $$\begin{aligned}&\int_{\mathbb{R}}\mathbbm{1}_{\{\rho\leq 1\}}(\partial_{x} \rho)^{2}dx\leq\check{C}\int_{\mathbb{R}}\mathbbm{1}_{\{\rho\leq 1\}}\bigg|\frac{\partial_{x}p}{\sqrt{\rho}\rho^{\underline{\alpha}\check{\gamma}}}\bigg|^{2}dx-\tilde{C}\bigg(\int_{\mathbb{R}}\mathbbm{1}_{\{\rho\leq 1\}}\bigg|\frac{\partial_{x}p}{\sqrt{\rho}\rho^{\underline{\alpha}\check{\gamma}}}\bigg|^{2}dx \\ &\times \int_{\mathbb{R}}\mathbbm{1}_{\{\rho\leq 1\}}\bigg|\rho^{\check{\gamma}(1-\underline{\alpha})+\frac{1}{2}}\Big(\frac{\partial_{x}\mu}{\rho}\Big)\bigg|^{2} dx\bigg)^{\frac{1}{2}}+ C\int_{\mathbb{R}}\mathbbm{1}_{\{\rho\leq 1\}}\bigg|\rho^{\check{\gamma}(1-\underline{\alpha})+\frac{1}{2}}\Big(\frac{\partial_{x}\mu}{\rho}\Big)\bigg|^{2}dx \\& \qquad\qquad\qquad\quad \ \ \leq \check{C}_{0}\int_{\mathbb{R}}\mathbbm{1}_{\{\rho\leq 1\}}\bigg|\frac{\partial_{x}\psi(p)}{\sqrt{\rho}}\bigg|^{2}dx-\tilde{C}_{0}\bigg(\int_{\mathbb{R}}\mathbbm{1}_{\{\rho\leq 1\}}\bigg|\frac{\partial_{x}\psi(p)}{\sqrt{\rho}}\bigg|^{2}dx\\ &\times \bar{\varrho}^{\check{\gamma}(1-\underline{\alpha})+\frac{1}{2}}\int_{\mathbb{R}}\mathbbm{1}_{\{\rho\leq 1\}}|\rho^{-1}\partial_{x}\mu|^{2}dx\bigg)^{\frac{1}{2}}+C\overline{\varrho}^{\check{\gamma}(1-\underline{\alpha})+\frac{1}{2}}\int_{\mathbb{R}}\mathbbm{1}_{\{\rho\leq 1\}}|(\rho^{-1}\partial_{x}\mu)|^{2} dx \\ &\qquad\qquad\qquad\quad \ \ \leq C,\end{aligned}$$ which concludes the proof for $\rho\leq 1$.

For the case $\rho\geq 1$ we follow an almost identical calculation, except that now after applying  (\ref{rho2}), (\ref{rhopb}) and (\ref{rhomub}); (\ref{bal6}) becomes \begin{equation}\begin{aligned}\label{bal7} \partial_{x}\rho^{\sigma} & = \rho^{\sigma-1}\partial_{x}\rho \\ & = \rho^{\sigma-1} \left(\frac{\partial_{x}p(\rho,\mu)-\partial_{\mu}p(\rho,\mu)\partial_{x}\mu}{\partial_{\rho}p(\rho,\mu)}\right) \\ & \leq \rho^{\sigma-1} \left(\frac{\check{C}\partial_{x}p(\rho,\mu)}{\rho^{\check{\gamma}-1}}\right)-C\rho^{\sigma-\check{\gamma}}\partial_{\mu}p(\rho,\mu)\partial_{x}\mu\\ &\leq \check{C}\rho^{-\overline{\alpha}\hat{\gamma}-1/2}\partial_{x}p(\rho,\mu)-C\rho^{\sigma-\check{\gamma}+\hat{\gamma}+1}\Big(\frac{\partial_{x}\mu}{\rho}\Big). \end{aligned}\end{equation}  Squaring both sides now gives $$\begin{aligned}(\partial_{x}\rho)^{2} & \leq \rho^{2-2\sigma}\left(\check{C}\rho^{-\overline{\alpha}\hat{\gamma} -1/2} \partial_{x}p(\rho,\mu)-C\rho^{\hat{\gamma}(1-\overline{\alpha})+\frac{1}{2}}\Big(\frac{\partial_{x}\mu}{\rho}\Big)\right)^{2}.\end{aligned}$$  Again integrating and applying (\ref{rho2}) with H\"{o}lder's inequality establishes, $$\begin{aligned}&\int_{\mathbb{R}}\mathbbm{1}_{\{\rho\geq 1\}}(\partial_{x} \rho)^{2}dx\leq\hat{C}\int_{\mathbb{R}}\mathbbm{1}_{\{\rho\geq 1\}}\bigg|\frac{\partial_{x}p}{\sqrt{\rho}\rho^{\overline{\alpha}\hat{\gamma}}}\bigg|^{2}dx-\tilde{C}\bigg(\int_{\mathbb{R}}\mathbbm{1}_{\{\rho\geq 1\}}\bigg|\frac{\partial_{x}p}{\sqrt{\rho}\rho^{\overline{\alpha}\hat{\gamma}}}\bigg|^{2}dx \\ &\times \int_{\mathbb{R}}\mathbbm{1}_{\{\rho\geq 1\}}\bigg|\rho^{\hat{\gamma}(1-\overline{\alpha})+\frac{1}{2}}\Big(\frac{\partial_{x}\mu}{\rho}\Big)\bigg|^{2} dx\bigg)^{\frac{1}{2}}+ C\int_{\mathbb{R}}\mathbbm{1}_{\{\rho\geq 1\}}\bigg|\rho^{\hat{\gamma}(1-\overline{\alpha})+\frac{1}{2}}\Big(\frac{\partial_{x}\mu}{\rho}\Big)\bigg|^{2}dx \\& \qquad\qquad\qquad\quad \ \ \leq \hat{C}_{0}\int_{\mathbb{R}}\mathbbm{1}_{\{\rho\geq 1\}}\bigg|\frac{\partial_{x}\psi(p)}{\sqrt{\rho}}\bigg|^{2}dx-\tilde{C}_{0}\bigg(\int_{\mathbb{R}}\mathbbm{1}_{\{\rho\geq 1\}}\bigg|\frac{\partial_{x}\psi(p)}{\sqrt{\rho}}\bigg|^{2}dx\\ &\times \bar{\varrho}^{\hat{\gamma}(1-\overline{\alpha})+\frac{1}{2}}\int_{\mathbb{R}}\mathbbm{1}_{\{\rho\geq 1\}}|\rho^{-1}\partial_{x}\mu|^{2}dx\bigg)^{\frac{1}{2}}+C\overline{\varrho}^{\hat{\gamma}(1-\overline{\alpha})+\frac{1}{2}}\int_{\mathbb{R}}\mathbbm{1}_{\{\rho\geq 1\}}|(\rho^{-1}\partial_{x}\mu)|^{2} dx \\ &\qquad\qquad\qquad\quad \ \ \leq C,\end{aligned}$$ which due to Minkowski's inequality completes the proof.
\end{proof}  

\subsection{4.2 Bounds for the Velocity}

It is now possible to find bounds on the velocity by applying the uniform bounds achieved above.    

\newtheorem*{u1}{Proposition 4.2}
\begin{u1} Assume that (\ref{visc})-(\ref{rho2}) and (\ref{rhopb})-(\ref{rhomub}) are satisfied, then \begin{equation}\label{usobo} u\in L^{2}(0,T;H^{2}(\mathbb{R})) \quad\mathit{and}\quad \partial_{t}u \in L^{2}(0,T;L^{2}(\mathbb{R})).\end{equation}
\end{u1}

\begin{proof} First notice that the second estimate in (\ref{apri1}) in tandem with the uniform bounds on the density gives \begin{equation}\label{ubo1}\|u\|_{L^{\infty}(0,T;L^{2}(\mathbb{R}))}\leq C.\end{equation} Also notice that the uniform bounds on $\rho$ applied to (\ref{visc}) show that there exists a constant $C$ such that $\nu(\rho,\mu)^{-1}\leq C$.  That is, applying (\ref{rho2}) and (\ref{rhopb}) to (\ref{visc}) for $\rho\leq 1$ gives $\nu(\rho,\mu)\geq C\rho^{\hat{\gamma}-\underline{\alpha}\check{\gamma}}$ so that using the uniform bounds on $\rho$ provides \begin{equation}\label{nuinv1}\nu(\rho,\mu)^{-1}\leq C\rho^{\underline{\alpha}\check{\gamma}-\hat{\gamma}}\leq C\underline{\varrho}^{\underline{\alpha}\check{\gamma}-\hat{\gamma}}\leq C.\end{equation}  For $\rho\geq 1$ it follows in the same way that $\nu(\rho,\mu)\geq C\rho^{\check{\gamma}-\overline{\alpha}\hat{\gamma}}$ provides \begin{equation}\label{nuinv2}\nu(\rho,\mu)^{-1}\leq C\rho^{\overline{\alpha}\hat{\gamma}-\check{\gamma}}\leq C\underline{\varrho}^{\overline{\alpha}\hat{\gamma}-\check{\gamma}}\leq C.\end{equation}  Thus for all $\rho$ we have $\nu(\rho,\mu)^{-1}\leq C$, which when applied to (\ref{apri1}) yields \begin{equation}\label{ubo2}\|u\|_{L^{2}(0,T;H^{1}(\mathbb{R})}\leq C.\end{equation}  Further, observing the continuity equation with respect to (\ref{ubo2}) implies that $\partial_{t}\rho$ is bounded in $L^{2}((0,T)\times \mathbb{R})$ as denoted in the theorem.  
We proceed by controlling the following form of the momentum equation (after multiplication through by $\rho^{-1}$): \begin{equation}\label{momexpand}\partial_{t}u-\partial_{x}\left(\rho^{-1}\nu(\rho,\mu)\partial_{x}u\right)=-u\partial_{x}u-\rho^{-1}\partial_{x}p(\rho,\mu)-\nu(\rho,\mu)\partial_{x}u\partial_{x}\rho^{-1}.\end{equation}  We want to control the right side of (\ref{momexpand}) in such a way as to apply classical regularity results for parabolic equations. 

Consider first the second term on the right in (\ref{momexpand}).  This term is bounded in $L^{\infty}(0,T;L^{2}(\mathbb{R}))$ as an immediate consequence of proposition 4.1, the corollary, and condition (\ref{rho2}).  This follows since (\ref{rho2}) gives $p(\rho,\mu)\leq C\rho^{\hat{\gamma}}$ for $\rho\geq 1$ and $p(\rho,\mu)\leq C\rho^{\check{\gamma}}$ for $\rho\leq 1$.  Then we can expand the pressure term as $\rho^{\hat{\gamma}-2}\partial_{x}\rho$ and $\rho^{\check{\gamma}-2}\partial_{x}\rho$, such that for $\rho\geq 1$ the corollary and proposition 4.1 provide that $$\begin{aligned}\int_{\mathbb{R}}\mathbbm{1}_{\{\rho\geq 1\}}|\rho^{-1}\partial_{x}p(\rho,\mu)|^{2}dx &\leq C\left(\esssup_{\{x\in\mathbb{R}:\rho\geq 1\}}|\rho^{2\hat{\gamma}-4}|\right)\left(\int_{\mathbb{R}}\mathbbm{1}_{\{\rho\geq 1\}}|\partial_{x}\rho|^{2} dx\right)\\ &\leq C,\end{aligned}$$ and likewise for $\rho\leq 1$ the corollary and proposition 4.1 give $$\begin{aligned}\int_{\mathbb{R}}\mathbbm{1}_{\{\rho\leq 1\}}|\rho^{-1}\partial_{x}p(\rho,\mu)|^{2}dx &\leq C\left(\esssup_{\{x\in\mathbb{R}:\rho\leq 1\}}|\rho^{2\check{\gamma}-4}|\right)\left(\int_{\mathbb{R}}\mathbbm{1}_{\{\rho\leq 1\}}|\partial_{x}\rho|^{2} dx\right)\\ &\leq C.\end{aligned}$$  Minkowski's inequality then provides the result.

For the third term on the right we again use the fact from above that $\nu(\rho,\mu)^{-1}\leq C$, and so because of the uniform bounds on $\rho$ we acquire $$|\nu(\rho,\mu)\partial_{x}u \partial_{x}\rho^{-1}| \leq C| \partial_{x}u\partial_{x}\rho|.$$  Hence, due to results on parabolic equations (see \cite{LSU}) we have reduced the problem to finding for the third term on the right in (\ref{momexpand}) that $\rho_{x}u_{x}$ is bounded in $L^{2}(0,T;L^{4/3}(\mathbb{R}))$ and similarly for the first term on the right that $uu_{x}$ is in $L^{2}(0,T;L^{4/3}(\mathbb{R}))$.  To get this, we adapt a subtle calculation from \cite{MV3} that relies on correctly weighting the norms in order to establish that $u_{x}\in L^{2}(0,T;L^{\infty}(\mathbb{R}))$.  That is, using H\"{o}lder's inequality we can write: \begin{equation}\begin{aligned}\label{suppa}\|uu_{x}\|&_{L^{2}(0,T;L^{4/3}(\mathbb{R}))}+\|\rho_{x}u_{x}\|_{L^{2}(0,T;L^{4/3}(\mathbb{R}))}\\ &\leq \big\{\|u\|_{L^{\infty}(0,T;L^{2}(\mathbb{R}))}+\|\rho_{x}\|_{L^{\infty}(0,T;L^{2}(\mathbb{R}))}\big\}\|u_{x}\|_{L^{2}(0,T;L^{4}(\mathbb{R}))}.\end{aligned}\end{equation} 
 
Now for some function $f$ with constant $a\in\mathbb{R}$ we have $(f^{a})_{x}=af^{a-1}f_{x}$ such that we may infer by H\"{o}lder's inequality that \begin{equation}\label{help}\|\partial_{x}(f^{3/2})\|_{L^{1}(\mathbb{R})}\leq C\|f^{1/2}\|_{L^{4}(\mathbb{R})}\|f_{x}\|_{L^{4/3}(\mathbb{R})}.\end{equation}  Next we infer a bound in $L^{8/3}(\mathbb{R})$ given by $$\|f^{3/2}\|_{L^{8/3}(\mathbb{R})}\leq C\|f^{3/2}\|^{1/2}_{L^{4/3}(\mathbb{R})}\|\partial_{x}(f^{3/2})\|^{1/2}_{L^{1}(\mathbb{R})},$$ which follows since $$\|\partial_{x}(f^{3/2})\|^{1/2}_{L^{1}(\mathbb{R})}\geq C\|f^{3/2}\|^{1/2}_{L^{\infty}(\mathbb{R})}.$$ Thus invoking (\ref{help}) we can write $$\begin{aligned}\|f\|^{3/2}_{L^{4}(\mathbb{R})} &\leq C\|f\|^{3/4}_{L^{2}(\mathbb{R})}\|f_{x}\|^{1/2}_{L^{4/3}(\mathbb{R})}\|\sqrt{f}\|^{1/2}_{L^{4}(\mathbb{R})}\\ &\leq C\|f\|_{L^{2}(\mathbb{R})}\|f_{x}\|^{1/2}_{L^{4/3}(\mathbb{R})},\end{aligned}$$ where both sides raised to the power $n=2/3$ clearly implies that $$\|f\|_{L^{4}(\mathbb{R})}\leq C\|f\|^{2/3}_{L^{2}(\mathbb{R})}\|f\|^{1/3}_{W^{1,4/3}}.$$  

Hence, if we set $u_{x}=f$ then (\ref{suppa}) leads to $$\begin{aligned}\big\{\|u\|&_{L^{\infty}(0,T;L^{2}(\mathbb{R}))}+\|\rho_{x}\|_{L^{\infty}(0,T;L^{2}(\mathbb{R}))}\big\}\|u_{x}\|_{L^{2}(0,T;L^{4}(\mathbb{R}))} \\ &\leq C\|u\|_{L^{\infty}(0,T;L^{2}(\mathbb{R}))}\|u_{x}\|^{2/3}_{L^{2}(0,T;L^{2}(\mathbb{R}))}\|u_{x}\|^{1/3}_{L^{2}(0,T;W^{1,4/3}(\mathbb{R}))} \\  &\quad +C\|\rho_{x}\|_{L^{\infty}(0,T;L^{2}(\mathbb{R}))}\|u_{x}\|^{2/3}_{L^{2}(0,T;L^{2}(\mathbb{R}))}\|u_{x}\|^{1/3}_{L^{2}(0,T;W^{1,4/3}(\mathbb{R}))}\\ &\leq C\|u_{x}\|^{1/3}_{L^{2}(0,T;W^{1,4/3}(\mathbb{R}))},\end{aligned}$$ since $u$ and $\rho_{x}$ are given by (\ref{ubo2}) and the corollary.  But then regularity results (see theorem 4.2 in Chapter III of \cite{LSU}) for equations of the form (\ref{momexpand}), given the bounds established above and that $\nu(\rho,\mu)$ is a coefficient function satisfying uniform parabolicity, imply that since $$\|\partial_{x}u\|_{L^{2}(0,T;W^{1,4/3}(\mathbb{R}))}\leq C + C\|u_{x}\|^{1/3}_{L^{2}(0,T;W^{1,4/3}(\mathbb{R}))},$$ we have \begin{equation}\label{fourthird}\|\partial_{x}u\|_{L^{2}(0,T;W^{1,4/3}(\mathbb{R}))}\leq C.\end{equation}   

Now, we want to show that \begin{equation}\label{udiv}u_{x}\in L^{2}(0,T;L^{\infty}(\mathbb{R})).\end{equation}  Indeed for any $x\in\mathbb{R}$ and $t\in [0,T]$ if we set $\varsigma=u_{x}$ from lemma 4.5 (which is given following this proof) and notice that $$\|u_{x}(t,x)\|^{2}\leq 2\|u_{x}(t,\cdot)\|^{2}_{L^{2}(\mathbb{R})}+\|u_{xx}(t,\cdot)\|^{2}_{L^{4/3}(\mathbb{R})}$$ for any $t\in [0,T]$, then integrating in time gives (\ref{udiv}).

It follows as a consequence that the entire right hand side of (\ref{momexpand}) is bounded in $L^{2}(0,T;L^{2}(\mathbb{R}))$.  Applying the classical regularity results for parabolic equations then yields: $$\|u\|_{L^{2}(0,T;H^{2}(\mathbb{R}))}\leq C\qquad\mathrm{and}\qquad\|\partial_{t}u\|_{L^{2}(0,T;L^{2}(\mathbb{R}))}\leq C.$$ \end{proof}

\newtheorem*{bind}{Lemma 4.5}
\begin{bind}
Let $\varsigma\in L^{2}(\mathbb{R})$ with $\partial_{x}\varsigma\in L^{1}_{loc}(\mathbb{R})$.  Then for any $x\in\mathbb{R}$ $$|\varsigma(x)|^{2}\leq 2\|\varsigma\|_{L^{2}(\mathbb{R})}^{2}+2\left(\int_{I}|\partial_{x}\varsigma|dz\right)^{2},$$ where $I=[x,x+1]$.
\end{bind}

\begin{proof}  It follows by the fundamental theorem that $$|\varsigma(x)| \leq |\varsigma(y)|+\int_{x}^{y}|\partial_{x}\varsigma|dz \leq |\varsigma(y)|+\int_{I}|\partial_{x}\varsigma|dz,$$ for any $y\in I$.  Squaring both sides and integrating over $\mathbb{R}$ in $y$ yields: $$|\varsigma(x)|^{2}\leq 2 \|\varsigma\|^{2}_{L^{2}(\mathbb{R})}+2\left(\int_{I}|\partial_{x}\varsigma|dz\right)^{2}.$$\end{proof}

\subsection{4.3 Bounds on the Mass Fraction}

All that remains in order to conclude the proof of the existence half of the theorem is to establish the bounds on $\mu$.  However, this is now an easy consequence of the bounds we have already established above.

\newtheorem*{mu1}{Lemma 4.6}
\begin{mu1} Given proposition 4.1 and 4.2 there exist constants such that, $$\|\mu_{x}\|_{L^{\infty}(0,T;L^{\infty}(\mathbb{R}))} \leq C\quad \mathit{and}\quad \|\partial_{t}\mu\|_{L^{\infty}(0,T;L^{2}(\mathbb{R}))}\leq C.$$
\end{mu1}

\begin{proof}  We have from lemma 4.4 that \begin{equation}\|\rho^{-1/2}\partial_{x}\mu\|_{L^{\infty}(0,T;L^{\infty}(\mathbb{R}))}\leq C,\end{equation} and so thanks to the uniform bounds on the density from Proposition 4.1, this yields that $\partial_{x}\mu$ is in $L^{\infty}(0,T;L^{\infty}(\mathbb{R}))$.  Now using (\ref{ubo1}) and the above with (\ref{ttt}) we find that $\partial_{t}\mu$ is in $L^{\infty}(0,T;L^{2}(\mathbb{R}))$.  \end{proof}

\subsection{4.4 Proof of the Existence Half of the Theorem}

We now apply the preceeding results in \textsection{4} in order to prove the existence theorem.

\begin{proof} [Proof of existence half of the theorem] In view of the \emph{a priori} estimates that we have now, the only difficulty that remains is to deal with the fact that $\nu$ is not uniformly bounded by below with respect to $\rho$.  This is needed to apply the short-existence result of Solonnikov (proposition 2.1).  To solve this problem let us fix any $T>0$.  Then we define an approximation to $\nu$ by, $$\tilde{\nu}(y,z)=\begin{cases}\nu(y,z) \ & \mathrm{if}\quad y\geq \frac{\underline{\varrho}(T)}{2}\\  \nu\left(\frac{\underline{\varrho}(T)}{2},z\right) \ & \mathrm{if}\quad y\leq \frac{\underline{\varrho}(T)}{2}\end{cases}$$ where $\underline{\varrho}(T)$ is defined by proposition 4.1.  Now let $(\tilde{\rho},\tilde{u},\tilde{\mu})$ be a strong solution of (\ref{rrr})-(\ref{ttt}), where $\nu$ is replaced by $\tilde{\nu}$; giving $$\begin{aligned}\partial_{t}\tilde{\rho}+\partial_{x}&(\tilde{\rho} \tilde{u})=0,\\ \partial_{t}(\tilde{\rho} \tilde{u})+\partial_{x}(\tilde{\rho}\tilde{u}^{2})+\partial_{x}p(\tilde{\rho},&\tilde{\mu})-\partial_{x}(\tilde{\nu}(\tilde{\rho},\tilde{\mu})\partial_{x}\tilde{u})=0, \\ \partial_{t}(\tilde{\rho}\tilde{\mu})+\partial_{x}&(\tilde{\rho}\tilde{u}\tilde{\mu})=0. \end{aligned}$$  By (\ref{visc}), (\ref{rho2}), and (\ref{rhopb}) the approximate function $\tilde{\nu}$ is bounded from below, thus proposition 2.1 provides that such a solution exists for all $t\in (0,T_{s})$.  Consider $\tilde{T}\leq T$ the biggest time such that $$\inf_{x}(\tilde{\rho}(t,\cdot))\geq \frac{\underline{\varrho}(T)}{2}.$$  Then on $[0,\tilde{T}]$, it follows that $\tilde{\nu}=\nu$.  Now assume that $\tilde{T}<T$.  From proposition 4.1, on $[0,\tilde{T}]$ $$\inf_{x}\tilde{\rho}(t,\cdot)\geq \underline{\varrho}(T)>\frac{\underline{\varrho}(T)}{2},$$ which contradicts the fact that $\tilde{T}<T$.  Hence we have constructed the solution of (\ref{rrr})-(\ref{ttt}) up to time $T$, and this for any $T>0$, which completes the proof.\end{proof}

\section{\protect\centering $\S 5$ Establishing the Uniqueness Theorem} 

Now we address the uniqueness half of the theorem.  Thanks to \cite{S1} this result follows fairly directly.  

\newtheorem*{th2}{Theorem}
\begin{th2} 
Let  $\psi '' (p)$, $\partial_{\rho\rho}p(\rho,\mu)$ and $\partial_{\rho\mu}p(\rho,\mu)$ be locally bounded.  Then a solution of (\ref{rrr})-(\ref{ttt}) verifying proposition 4.1, proposition 4.2, and lemma 4.6 is uniquely determined.
\end{th2}

\begin{proof}  Let $(\rho_{1}, u_{1},\mu_{1})$ and $(\rho_{2},u_{2},\mu_{2})$ be two solutions to the system (\ref{rrr})-(\ref{ttt}), and define $\chi=\mu_{1}-\mu_{2}$, $\tau=\rho_{1}-\rho_{2}$, $\zeta=u_{1}-u_{2}$, $p_{\ell}=p(\rho_{1},\mu_{1})-p(\rho_{2},\mu_{2})$ and $\nu_{\ell}=\nu(\rho_{1},\mu_{1})-\nu(\rho_{2},\mu_{2})$ such that from (\ref{rrr})-(\ref{ttt}) we can write: $$\begin{aligned}&\partial_{t}\tau + \partial_{x}(\rho_{1}u_{1}-\rho_{2}u_{2})=0, \\  \rho_{1}\partial_{t}u_{1}-\rho_{2}\partial_{t} u_{2} +\rho_{1} u_{1}&\partial_{x}u_{1}-\rho_{2}u_{2}\partial_{x}u_{2} +\partial_{x}p_{\ell} - \partial_{x}(\nu_{1}\partial_{x}u_{1}-\nu_{2}\partial_{x}u_{2}) = 0, \\ \partial_{t} &\chi+ (u_{1}\partial_{x}\mu_{1} -  u_{2} \partial_{x} \mu_{2}) = 0.\end{aligned}$$  By rearranging we get \begin{align}\label{rrr3}& \quad \partial_{t}\tau +\partial_{x}(\tau u_{1}+\rho_{2}\zeta)=0,\\ \rho_{1}(\partial_{t}\zeta + u_{1}&\partial_{x}\zeta+\zeta\partial_{x} u_{2})+\tau(\partial_{t}u_{2}+u_{2}\partial_{x}u_{2}) \notag \\  \label{sss3} & \quad\quad\quad +\partial_{x}p_{\ell}-\partial_{x}(\nu_{\ell}\partial_{x}u_{1})-\partial_{x}(\nu_{2}\partial_{x}\zeta)=0,\\ \label{ttt3} & \quad\partial_{t}\chi+\zeta\partial_{x}\mu_{1}+u_{2}\partial_{x}\chi =0. \end{align}    

First let us consider equation (\ref{rrr3}).  Here we multiply through by $\tau$ and integrate in $x$.  To begin with, note that the first term on the left satisfies \begin{equation}\label{mass1}\int_{\mathbb{R}}\tau\partial_{t}\tau dx=\frac{1}{2}\int_{\mathbb{R}}\partial_{t}\tau^{2}dx.\end{equation} For the $(\tau u_{1})_{x}$ term we use proposition 4.2 as applied in (\ref{udiv}) by setting $u=u_{1}$ to see that \begin{equation}\label{mass2}\Big|\int_{\mathbb{R}}\tau (\tau u_{1})_{x}dx\Big|\leq \frac{1}{2}\|\tau^{2}\partial_{x}u_{1}\|_{L^{1}(\mathbb{R})} \leq \frac{1}{2}\|\tau\|^{2}_{L^{2}(\mathbb{R})}\|\partial_{x}u_{1}\|_{L^{\infty}(\mathbb{R})} \leq B_{1}(t)\|\tau\|^{2}_{L^{2}(\mathbb{R})}.\end{equation}  For the $(\rho_{2}\zeta)_{x}$ term notice that we can write: $$\begin{aligned}&\Big|\int_{\mathbb{R}}\tau\partial_{x}(\rho_{2}\zeta)dx\Big|\leq\Big|\int_{\mathbb{R}}\tau\rho_{2}\partial_{x}\zeta dx\Big|+\Big|\int_{\mathbb{R}}\tau\zeta\partial_{x}\rho_{2}dx \Big|.\end{aligned}$$  Applying proposition 4.1 and Cauchy's inequality to the first term on the right provides, \begin{equation}\begin{aligned}\label{mass3}\Big|\int_{\mathbb{R}}\tau\rho_{2}\partial_{x}\zeta dx \Big|& \leq C \int_{\mathbb{R}}|\tau \partial_{x}\zeta| dx \leq C\|\tau\|_{L^{2}(\mathbb{R})}\|\zeta_{x}\|_{L^{2}(\mathbb{R})} \\ & \leq  C^{2}(4\epsilon_{1})^{-1}\|\tau\|^{2}_{L^{2}(\mathbb{R})}+ \epsilon_{1}\|\zeta_{x}\|^{2}_{L^{2}(\mathbb{R})}.\end{aligned}\end{equation}  For the second term on the right H\"{o}lder's inequality with the corollary implies that \begin{equation}\begin{aligned}\label{mass4}\Big|\int_{\mathbb{R}} \tau\zeta\partial_{x}\rho_{2}dx \Big| & \leq \left(\int_{\mathbb{R}}|\tau|^{2} dx\right)^{1/2}\left(\int_{\mathbb{R}}|\partial_{x}\rho_{2}|^{2}dx\right)^{1/2} \left(\esssup_{\mathbb{R}}|\zeta|\right) \\ & \leq C\left(\int_{\mathbb{R}}|\tau|^{2} dx\right)^{1/2}\left(\esssup_{\mathbb{R}}|\zeta|\right).\end{aligned}\end{equation}  Now we utilize lemma 4.5 by setting $\varsigma=\zeta$.  Since $|\zeta|\leq |u_{1}|+|u_{2}|$ the bounds in (\ref{ubo1}) provide that $\zeta\in L^{\infty}(0,T;L^{2}(\mathbb{R}))$.  Furthermore, proposition 4.2 gives that since $|\zeta_{x}|^{2}\leq 2|\partial_{x}u_{1}|^{2}+2|\partial_{x}u_{2}|^{2}$ we have $\zeta_{x}\in L^{2}(0,T;L^{2}(\mathbb{R}))$.  Thus noticing that $\|\zeta\|_{L^{1}(I)}\leq \|\zeta\|_{L^{2}(I)}$ since $|I|=1$ from lemma 4.5, it follows that $$|\zeta(x)|\leq \|\zeta\|_{L^{2}(\mathbb{R})}+\|\zeta_{x}\|_{L^{1}(I)}\leq \|\zeta\|_{L^{2}(\mathbb{R})}+\|\zeta_{x}\|_{L^{2}(\mathbb{R})},$$ allowing us to deduce, $$\|\zeta\|_{L^{\infty}(\mathbb{R})}\leq \|\zeta\|_{L^{2}(\mathbb{R})}+\|\partial_{x}\zeta\|_{L^{2}(\mathbb{R})}.$$ By Cauchy's inequality this finally yields \begin{equation}\begin{aligned}\label{mass5} C\|\tau&\|_{L^{2}(\mathbb{R})}\|\zeta\|_{L^{\infty}(\mathbb{R})} \\ &\leq \epsilon_{2} \|\zeta_{x}\|^{2}_{L^{2}(\mathbb{R})}+ \bigg\{\frac{C^{2}}{4\epsilon_{2}}+\frac{C}{2}\bigg\}\left(\|\tau\|^{2}_{L^{2}(\mathbb{R})}+\|\zeta\|^{2}_{L^{2}(\mathbb{R})}\right).\end{aligned}\end{equation}  Thus combining (\ref{mass1}), (\ref{mass2}) and (\ref{mass5}) allows us to write for (\ref{rrr3}): \begin{equation}\begin{aligned}\label{mass6} \frac{1}{2}\frac{d}{dt}\int_{\mathbb{R}}&\tau^{2}dx -\big\{\epsilon_{1}+\epsilon_{2} \big\} \int_{\mathbb{R}}(\partial_{x}\zeta)^{2} dx \\& \leq \bigg\{B_{1}(t)+\frac{C^{2}}{4\epsilon_{1}}+\frac{C^{2}}{4\epsilon_{2}}+\frac{C}{2}\bigg\}\left(\|\tau\|^{2}_{L^{2}(\mathbb{R})}+\|\zeta\|^{2}_{L^{2}(\mathbb{R})}\right).\end{aligned}\end{equation}  

Next we want to multiply (\ref{sss3}) through by $\zeta$ and integrate in $\mathbb{R}$.  For the first two terms in the first part of (\ref{sss3}) we find: \begin{equation}\begin{aligned}\label{moment1}\int_{\mathbb{R}}\rho_{1}\zeta (\partial_{t}\zeta + u_{1}\partial_{x}\zeta)dx & = \int_{\mathbb{R}}\frac{\rho_{1}}{2}(\partial_{t}\zeta^{2}+u_{1}\partial_{x}\zeta^{2})dx \\ & = \frac{1}{2}\frac{d}{dt}\int_{\mathbb{R}}\rho_{1}\zeta^{2} dx - \int_{\mathbb{R}}\frac{\zeta^{2}}{2}(\partial_{t}\rho_{1}+\partial_{x}(\rho_{1}u_{1}))dx \\ & = \frac{1}{2}\frac{d}{dt}\int_{\mathbb{R}}\rho_{1}\zeta^{2} dx.\end{aligned}\end{equation}   For the $\rho_{1}\zeta\partial_{x}u_{2}$ term in (\ref{sss3}) we use the same calculation given in (\ref{mass2}) which is formulated in (\ref{udiv}) by setting $u=u_{2}$ such that, \begin{equation}\label{moment2} \Big|\int_{\mathbb{R}}\rho_{1}\zeta^{2}\partial_{x}u_{2}dx\Big| \leq C\|\zeta\|^{2}_{L^{2}(\mathbb{R})}\|\partial_{x}u_{2}\|_{L^{\infty}(\mathbb{R})} \leq B_{2}(t)\|\zeta\|^{2}_{L^{2}(\mathbb{R})}.\end{equation} Now, for the $\tau(\partial_{t}u_{2}+u_{2}\partial_{x}u_{2})$ part of (\ref{sss3}) we utilize a calculation similar to that employed for the term in (\ref{mass4}).  Here we simply substitute the $\partial_{x}\rho_{2}$ term from (\ref{mass4}) with $\omega = \partial_{t}u_{2} + u_{2}\partial_{x}u_{2}$, noting that proposition 4.2 along with (\ref{udiv}) assure that $\omega$ is bounded in $L^{2}(0,T;L^{2}(\mathbb{R}))$.  Thus we obtain \begin{equation}\begin{aligned}\label{moment3} B(t)\|\tau\|_{L^{2}(\mathbb{R})}&\|\zeta\|_{L^{\infty}(\mathbb{R})}\leq \epsilon_{3}\|\zeta_{x}\|^{2}_{L^{2}(\mathbb{R})}+ \epsilon_{3}^{-1}B_{3}(t)\left(\|\tau\|^{2}_{L^{2}(\mathbb{R})}+\|\zeta\|^{2}_{L^{2}(\mathbb{R})}\right),\end{aligned}\end{equation} where here $B_{3}(t)=\epsilon_{3}B(t)/2+B(t)^{2}/4$.

Next consider the pressure term $p_{\ell}$ in (\ref{sss3}).  Here set $$\int_{\mathbb{R}} \zeta \partial_{x}p_{\ell}dx = -\int_{\mathbb{R}}\Big\{p(\rho_{1},\mu_{1})-p(\rho_{2},\mu_{2})\Big\}\partial_{x}\zeta dx.$$  The uniform bounds on $\rho$ along with (\ref{rhopb}) and (\ref{rhomub}) give that $|\partial_{\rho}p(\rho,\mu)|\leq C$ and $|\partial_{\mu}p(\rho,\mu)|\leq C$, and so $$|p(\rho_{2},\mu_{2})-p(\rho_{1},\mu_{1})|\leq C(|\tau|+|\chi|).$$  Thus $$\int_{\mathbb{R}} \zeta \partial_{x}p_{\ell}dx \leq C\int_{\mathbb{R}}(|\tau|+|\chi|)\partial_{x}\zeta dx,$$ which gives by Cauchy's inequality, \begin{equation}\label{moment4}\int_{\mathbb{R}} \zeta \partial_{x}p_{\ell}dx \leq 2\epsilon_{4}\int_{\mathbb{R}}(\partial_{x}\zeta)^{2}dx+\frac{C^{2}}{4\epsilon_{4}}\int_{\mathbb{R}}|\tau|^{2}dx+\frac{C^{2}}{4\epsilon_{4}}\int_{\mathbb{R}}|\chi|^{2}dx.\end{equation} 

Finally we consider the viscosity terms in (\ref{sss3}).  For the  $(\nu_{\ell}\partial_{x}u_{1})_{x}$ term $$-\int_{\mathbb{R}}\zeta\partial_{x}(\nu_{\ell}\partial_{x}u_{1})dx = \int_{\mathbb{R}}\nu_{\ell}\partial_{x}\zeta\partial_{x}u_{1}dx.$$  Since $\psi '' (p)$, $\partial_{\rho\rho}p(\rho,\mu)$ and $\partial_{\rho\mu}p(\rho,\mu)$ are locally bounded, then from (\ref{visc}) we have $\nu_{\ell}\leq C(|\tau|+|\chi|)$, which gives  $$-\int_{\mathbb{R}}\zeta\partial_{x}(\nu_{\ell}\partial_{x}u_{1})dx\leq C\int_{\mathbb{R}}(|\tau|+|\chi|)\partial_{x}\zeta\partial_{x}u_{1} dx,$$ and leads to, $$-\int_{\mathbb{R}}\zeta(\nu_{\ell}\partial_{x}u_{1})_{x}dx\leq 2\epsilon_{5}\int_{\mathbb{R}}\zeta_{x}^{2}dx+\frac{C^{2}}{4\epsilon_{5}}\int_{\mathbb{R}}|\partial_{x}u_{1}|^{2}\tau^{2}dx+\frac{C^{2}}{4\epsilon_{5}}\int_{\mathbb{R}}|\partial_{x}u_{1}|^{2}\chi^{2}dx.$$  Next we again use the fact that $\partial_{x}u_{1}$ is bounded in $L^{2}(0,T;L^{\infty}(\mathbb{R}))$ by (\ref{udiv}).  It subsequently follows that, \begin{equation}\label{moment5}-\|\zeta(\nu_{\ell}\partial_{x}u_{1})_{x}\|_{L^{1}(\mathbb{R})}\leq 2\epsilon_{5}\|\partial_{x}\zeta\|^{2}_{L^{2}(\mathbb{R})}+\frac{B_{4}(t)}{2\epsilon_{5}}(\|\tau\|^{2}_{L^{2}(\mathbb{R})}+\|\chi\|^{2}_{L^{2}(\mathbb{R})}).\end{equation}  For the $(\nu_{2}\zeta_{x})_{x}$ term we simply multiply through by $\zeta$ and integrate, yielding \begin{equation}\label{moment6}-\int_{\mathbb{R}}\zeta\partial_{x}(\nu_{2}\partial_{x}\zeta)dx=\int_{\mathbb{R}}\nu_{2}(\partial_{x}\zeta)^{2}dx\geq C\int_{\mathbb{R}}(\partial_{x}\zeta)^{2}dx\end{equation} when using that $\nu_{2}\geq C$.

Hence combining (\ref{moment1})-(\ref{moment6}) we have: \begin{equation}\begin{aligned}\label{moment7} \frac{1}{2}&\frac{d}{dt}\int_{\mathbb{R}}\rho_{1}\zeta^{2} dx+(C-\epsilon_{3}-2\epsilon_{4}-2\epsilon_{5})\int_{\mathbb{R}}|\partial_{x}\zeta|^{2} dx \\ & \quad\quad \ \leq \Big\{B_{2}(t)+ \frac{B_{3}(t)}{\epsilon_{3}}+\frac{B_{4}(t)}{4\epsilon_{5}}+\frac{C^{2}}{4\epsilon_{4}}\Big\}\left(\|\tau\|^{2}_{L^{2}(\mathbb{R})}+\|\zeta\|^{2}_{L^{2}(\mathbb{R})}\right).\end{aligned}\end{equation}

All that is left is to find a compatible form of equation (\ref{ttt3}).  Here we multiply through by $\chi$ and integrate in $\mathbb{R}$ such that the first term gives \begin{equation}\label{transport1}\int_{\mathbb{R}}\chi\partial_{t}\chi dx = \frac{1}{2}\frac{d}{dt}\int_{\mathbb{R}}\chi^{2} dx.\end{equation}  The second term in (\ref{ttt3}) is treated in a similar way as (\ref{mass4}) and (\ref{moment3}), where here we have \begin{equation}\begin{aligned}\Big|\int_{\mathbb{R}} \chi\zeta\partial_{x}\mu_{1}dx \Big| & \leq \left(\int_{\mathbb{R}}|\chi|^{2} dx\right)^{1/2}\left(\int_{\mathbb{R}}|\zeta|^{2}dx\right)^{1/2} \left(\esssup_{\mathbb{R}}|\partial_{x}\mu_{1}|\right) \\ & \leq C\left(\int_{\mathbb{R}}|\chi|^{2} dx\right)^{1/2}\left(\int_{\mathbb{R}}|\zeta|^{2}\right)^{1/2}.\end{aligned}\end{equation}  Thus we obtain, \begin{equation}\begin{aligned}\label{transport2} C\|\chi\|_{L^{2}(\mathbb{R})}&\|\zeta\|_{L^{2}(\mathbb{R})}\leq \frac{C}{2}\left(\|\chi\|^{2}_{L^{2}(\mathbb{R})}+\|\zeta\|^{2}_{L^{2}(\mathbb{R})}\right).\end{aligned}\end{equation}  For the last term in (\ref{ttt3}) we use (\ref{udiv}) with $u=u_{2}$ to see \begin{equation}\label{transport3} \Big|\int_{\mathbb{R}} \chi u_{2}\partial_{x}\chi dx \Big| \leq C \|\chi\|^{2}_{L^{2}(\mathbb{R})}\|\partial_{x}u_{2}\|_{L^{\infty}(\mathbb{R})}\leq B_{5}(t)\|\chi\|^{2}_{L^{2}(\mathbb{R})}.\end{equation} Thus putting (\ref{transport1}), (\ref{transport2}) and (\ref{transport3}) together yields, \begin{equation}\label{transport4} \frac{1}{2}\frac{d}{dt}\int_{\mathbb{R}}\chi^{2} dx \leq\Big\{C/2+B_{5}(t)\Big\}\left(\|\chi\|^{2}_{L^{2}(\mathbb{R})}+\|\zeta\|^{2}_{L^{2}(\mathbb{R})}\right).\end{equation}

Finally, combining (\ref{mass6}), (\ref{moment6}) and (\ref{transport4}) along with defining, $$\begin{aligned} \mathscr{C}=C-\epsilon_{1}+\epsilon_{2}+&\epsilon_{3}+2\epsilon_{4}+2\epsilon_{5}, \\ \mathscr{B}_{1}(t)=B_{1}(t)+C^{2}(4\epsilon_{1})&^{-1}+C^{2}(4\epsilon_{2})^{-1}+C/2, \\ \mathscr{B}_{2}(t) = B_{2}(t)+B_{3}(t)(\epsilon_{3})^{-1}&+C^{2}(4\epsilon_{4})^{-1}+B_{4}(2\epsilon_{5})^{-1},\\ \mathscr{B}_{3}(t) =B_{5}&(t)+C/2, \\ \mathscr{A}(t)=\mathscr{B}_{1}(t)+&\mathscr{B}_{2}(t)+\mathscr{B}_{3}(t)\\ \mathscr{X}(t)=(\tau^{2}+&\rho_{1}\zeta^{2}+\chi^{2}),\end{aligned}$$ yields: $$\begin{aligned}\frac{1}{2}\frac{d}{dt}\int_{\mathbb{R}}\mathscr{X}(t)dx + \mathscr{C}\int_{\mathbb{R}}|\partial_{x}\zeta|^{2} dx\leq \mathscr{A}(t)\left(\|\chi\|^{2}_{L^{2}(\mathbb{R})}+\|\zeta\|^{2}_{L^{2}(\mathbb{R})}+\|\tau\|^{2}_{L^{2}(\mathbb{R})}\right).\end{aligned}$$   Since proposition 4.1, proposition 4.2 and lemma 4.6 confirm by above that $\mathscr{A}(t)\in L^{2}(0,T)$, and as $\mathscr{C}$ is positive, then at $t=0$ since $$\int_{\mathbb{R}}\mathscr{X}(t_{0})dx=\int_{\mathbb{R}}\tau_{0}^{2}+\rho_{1|t=o}\zeta_{0}^{2}+\chi_{0}^{2}dx=0,$$ then Gronwall's lemma gives that $\int_{\mathbb{R}}\mathscr{X}(t)dx \equiv 0$ over $[0,T]$, which establishes that $\tau$, $\zeta$, and $\chi$ are each zero. \end{proof}


\begin{thebibliography}{00}

\bibitem{AS}
Y.~Amirat and V.~Shelukhin.
\newblock Global weak solutions to equations of compressible miscible flows in
  porous media.
\newblock {\em SIAM J. Math. Anal.}, 38(6):1825--1846 (electronic), 2007.

\bibitem{AK}
S.~N. Antontsev and A.~V. Kazhikhov.
\newblock {\em Matematicheskie voprosy dinamiki neodnorodnykh zhidkostei}.
\newblock Novosibirsk. Gosudarstv. Univ., Novosibirsk, 1973.
\newblock Lecture notes, Novosibirsk State University.

\bibitem{AKM}
S.~N. Antontsev, A.~V. Kazhikhov, and V.~N. Monakhov.
\newblock {\em Boundary value problems in mechanics of nonhomogeneous fluids},
  volume~22 of {\em Studies in Mathematics and its Applications}.
\newblock North-Holland Publishing Co., Amsterdam, 1990.
\newblock Translated from the Russian.

\bibitem{Binder}
K.~Binder.
\newblock Collective diffusion, nucleation and spinodal decomposition in
  polymer mixtures.
\newblock {\em Journal of Chemical Physics}, 79(12):6387--6409, 1983.

\bibitem{BD4}
D.~Bresch and B.~Desjardins.
\newblock Sur un mod\`ele de {S}aint-{V}enant visqueux et sa limite
  quasi-g\'eostrophique.
\newblock {\em C. R. Math. Acad. Sci. Paris}, 335(12):1079--1084, 2002.

\bibitem{BD0}
D.~Bresch and B.~Desjardins.
\newblock Some diffusive capillary models of korteweg type.
\newblock {\em C.~R.~Acad.~Sci.,~Paris,~Section~ M{\'e}canique},
  332(11):881--886, 2004.

\bibitem{BD1}
D.~Bresch and B.~Desjardins.
\newblock On compressible {N}avier-{S}tokes equations with density dependent
  viscosities in bounded domains.
\newblock {\em J. Math. Pures Appl. (9)}, 87(2):227--235, 2007.

\bibitem{BD3}
D.~Bresch and B.~Desjardins.
\newblock On the existence of global weak solutions to the {N}avier-{S}tokes
  equations for viscous compressible and heat conducting fluids.
\newblock {\em J. Math. Pures Appl. (9)}, 87(1):57--90, 2007.

\bibitem{BDL}
D.~Bresch, B.~Desjardins, and C-K. Lin.
\newblock On some compressible fluid models: {K}orteweg, lubrication, and
  shallow water systems.
\newblock {\em Comm. Partial Differential Equations}, 28(3-4):843--868, 2003.

\bibitem{BD2}
D.~Bresch, B.~Desjardins, and G.~M{\'e}tivier.
\newblock Recent mathematical results and open problems about shallow water
  equations.
\newblock In {\em Analysis and simulation of fluid dynamics}, Adv. Math. Fluid
  Mech., pages 15--31. Birkh\"auser, Basel, 2007.

\bibitem{BMR}
J.~Buajarern, L.~Mitchem, and J.~P. Reid.
\newblock Characterizing multiphase organic/inorganic/aqueous aerosol droplets.
\newblock {\em Journal of Physical Chemistry A}, 111(37):9054--9061, 2007.

\bibitem{BMR2}
M.~Bul{\'{\i}}{\v{c}}ek, J.~M{\'a}lek, and K.~R. Rajagopal.
\newblock Navier's slip and evolutionary {N}avier-{S}tokes-like systems with
  pressure and shear-rate dependent viscosity.
\newblock {\em Indiana Univ. Math. J.}, 56(1):51--85, 2007.

\bibitem{CH}
J.W. Cahn and J.E. Hilliard.
\newblock Free energy of a nonuniform system i. interfacial free energy.
\newblock {\em Journal of Chemical Physics}, 30:258--267, 1959.

\bibitem{CHT}
G-Q. Chen, D.~Hoff, and K.~Trivisa.
\newblock Global solutions of the compressible {N}avier-{S}tokes equations with
  large discontinuous initial data.
\newblock {\em Comm. Partial Differential Equations}, 25(11-12):2233--2257,
  2000.

\bibitem{CHT2}
G-Q. Chen, D.~Hoff, and K.~Trivisa.
\newblock Global solutions to a model for exothermically reacting, compressible
  flows with large discontinuous initial data.
\newblock {\em Arch. Ration. Mech. Anal.}, 166(4):321--358, 2003.

\bibitem{ChenK}
G-Q. Chen and M.~Kratka.
\newblock Global solutions to the {N}avier-{S}tokes equations for compressible
  heat-conducting flow with symmetry and free boundary.
\newblock {\em Comm. Partial Differential Equations}, 27(5-6):907--943, 2002.

\bibitem{CT}
G-Q. Chen and K.~Trivisa.
\newblock Analysis on models for exothermically reacting, compressible flows
  with large discontinous initial data.
\newblock In {\em Nonlinear partial differential equations and related
  analysis}, volume 371 of {\em Contemp. Math.}, pages 73--91. Amer. Math.
  Soc., Providence, RI, 2005.

\bibitem{CK}
Y.~Cho and H.~Kim.
\newblock Existence results for viscous polytropic fluids with vacuum.
\newblock {\em J. Differential Equations}, 228(2):377--411, 2006.

\bibitem{DEP}
B.~Das, G.~Enden, and A.S. Popel.
\newblock Stratified multiphase model for blood flow in a venular bifurcation.
\newblock {\em Annals of Biomedical Engineering}, 25(1):135--153, 1997.

\bibitem{DS}
E.J. Davis and G.~Schweiger.
\newblock {\em The Airborne Microparticle}.
\newblock Springer-Verlag, 2002.

\bibitem{BD5}
B.~Desjardins.
\newblock Regularity results for two-dimensional flows of multiphase viscous
  fluids.
\newblock {\em Arch. Rational Mech. Anal.}, 137(2):135--158, 1997.

\bibitem{DL}
R.~J. DiPerna and P.-L. Lions.
\newblock Ordinary differential equations, transport theory and {S}obolev
  spaces.
\newblock {\em Invent. Math.}, 98(3):511--547, 1989.

\bibitem{DT2}
D.~Donatelli and K.~Trivisa.
\newblock On the motion of a viscous compressible radiative-reacting gas.
\newblock {\em Comm. Math. Phys.}, 265(2):463--491, 2006.

\bibitem{DT}
D.~Donatelli and K.~Trivisa.
\newblock A multidimensional model for the combustion of compressible fluids.
\newblock {\em Arch. Ration. Mech. Anal.}, 185(3):379--408, 2007.

\bibitem{DN}
L.K. Doraiswamy and S.D. Naik.
\newblock Phase transfer catalysis: Chemistry and engineering.
\newblock {\em AICHE Journal}, 44(3):612--646, 1998.

\bibitem{DF}
B.~Ducomet and E.~Feireisl.
\newblock On the dynamics of gaseous stars.
\newblock {\em Arch. Ration. Mech. Anal.}, 174(2):221--266, 2004.

\bibitem{DF2}
B.~Ducomet and E.~Feireisl.
\newblock The equations of magnetohydrodynamics: on the interaction between
  matter and radiation in the evolution of gaseous stars.
\newblock {\em Comm. Math. Phys.}, 266(3):595--629, 2006.

\bibitem{DZ}
B.~Ducomet and A.~Zlotnik.
\newblock On the large-time behavior of 1{D} radiative and reactive viscous
  flows for higher-order kinetics.
\newblock {\em Nonlinear Anal.}, 63(8):1011--1033, 2005.

\bibitem{Dukowicz}
J.K. Dukowicz.
\newblock A particle-fluid numerical-model for liquid sprays.
\newblock {\em Journal of Computational Physics}, 35(2):229--253, 1980.

\bibitem{Eisenbach}
M.~Eisenbach.
\newblock {\em Chemotaxis}.
\newblock Imperial College Press, 2004.

\bibitem{Evans}
E.~Evans.
\newblock New physical concpets for cell ameboid motion.
\newblock {\em Biophysical Journal}, 64(4):1306--1322, 1993.

\bibitem{Faeth}
G.M. Faeth.
\newblock Evaporation and combustion of sprays.
\newblock {\em Progress in Energy and Combustion Science}, 9(1--2):1--76, 1983.

\bibitem{Faeth2}
G.M. Faeth.
\newblock Mixing, transport and combustion in sprays.
\newblock {\em Progress in Energy and Combustion Science}, 14(4):293--345,
  1987.

\bibitem{F1}
E.~Feireisl.
\newblock On the motion of a viscous, compressible, and heat conducting fluid.
\newblock {\em Indiana Univ. Math. J.}, 53(6):1705--1738, 2004.

\bibitem{F2}
E.~Feireisl.
\newblock Mathematics of viscous, compressible, and heat conducting fluids.
\newblock In {\em Nonlinear partial differential equations and related
  analysis}, volume 371 of {\em Contemp. Math.}, pages 133--151. Amer. Math.
  Soc., Providence, RI, 2005.

\bibitem{F3}
E.~Feireisl.
\newblock Mathematical theory of compressible, viscous, and heat conducting
  fluids.
\newblock {\em Comput. Math. Appl.}, 53(3-4):461--490, 2007.

\bibitem{F5}
E.~Feireisl, P.~Lauren{\c{c}}ot, and H.~Petzeltov{\'a}.
\newblock On convergence to equilibria for the {K}eller-{S}egel chemotaxis
  model.
\newblock {\em J. Differential Equations}, 236(2):551--569, 2007.

\bibitem{FFS}
M.~Feistauer, J.~Felcman, and I.~Stra{\v{s}}kraba.
\newblock {\em Mathematical and computational methods for compressible flow}.
\newblock Numerical mathematics and scientific computation. Oxford University
  Press, 2003.

\bibitem{FW}
J.~Fine and L.~Waite.
\newblock {\em Applied Biofluid Mechanics}.
\newblock The McGraw-Hill Companies, 2007.

\bibitem{FMR}
M.~Franta, J.~M{\'a}lek, and K.~R. Rajagopal.
\newblock On steady flows of fluids with pressure- and shear-dependent
  viscosities.
\newblock {\em Proc. R. Soc. Lond. Ser. A Math. Phys. Eng. Sci.},
  461(2055):651--670, 2005.

\bibitem{FTKH}
M.~Fujimoto, T.~Kado, W.~Takashima, K.~Kaneto, and S.~Hayase.
\newblock Dye-sensitized solar cells fabricated by electrospray coating using
  tio2 nanocrystal dispersion solution.
\newblock {\em Journal of the Electrochemical Society}, 153(5):A826--A829,
  2006.

\bibitem{Gardiner}
W.C. Gardiner.
\newblock {\em Combustion Chemistry}.
\newblock Springer-Verlag New York Inc., New York, NY, 1984.

\bibitem{HA}
F.H. Harlow and A.A. Amsden.
\newblock Numerical-calculation of multiphase fluid-flow.
\newblock {\em Journal of Computational Physics}, 17(1):19--52, 1975.

\bibitem{HG}
R.M. Harrison and R.E. van Grieken.
\newblock {\em Atmospheric Particles}, volume~5 of {\em IUPAC Series on
  Analytical and Physical Chemistry of Environmental Systems}.
\newblock John Wiley \& Sons, New York, NY, 1998.

\bibitem{HJ}
S.Y. Heriot and R.A.L. Jones.
\newblock An interfacial instability in a transient wetting layer leads to
  lateral phase separation in thin spin-cast polymer-blend films.
\newblock {\em Nature Materials}, 4(10):782--786, 2005.

\bibitem{Hirsch}
C.~Hirsch.
\newblock {\em Numerical computation of internal and external flows}, volume
  1--2 of {\em Wiley series in numerical methods in engineering}.
\newblock John Wiley \& Sons Ltd., Chichester [England], 1988.

\bibitem{Hoff1}
D.~Hoff.
\newblock Global existence for {$1$}{D}, compressible, isentropic
  {N}avier-{S}tokes equations with large initial data.
\newblock {\em Trans. Amer. Math. Soc.}, 303(1):169--181, 1987.

\bibitem{Hoff2}
D.~Hoff.
\newblock Strong convergence to global solutions for multidimensional flows of
  compressible, viscous fluids with polytropic equations of state and
  discontinuous initial data.
\newblock {\em Arch. Rational Mech. Anal.}, 132(1):1--14, 1995.

\bibitem{Hoff4}
D.~Hoff.
\newblock Global solutions of the equations of one-dimensional, compressible
  flow with large data and forces, and with differing end states.
\newblock {\em Z. Angew. Math. Phys.}, 49(5):774--785, 1998.

\bibitem{HS}
D.~Hoff and J.~Smoller.
\newblock Non-formation of vacuum states for compressible {N}avier-{S}tokes
  equations.
\newblock {\em Comm. Math. Phys.}, 216(2):255--276, 2001.

\bibitem{HT}
D.~Hoff and E.~Tsyganov.
\newblock Uniqueness and continuous dependence of weak solutions in
  compressible magnetohydrodynamics.
\newblock {\em Z. Angew. Math. Phys.}, 56(5):791--804, 2005.

\bibitem{Houston}
P.L. Houston.
\newblock {\em Chemical Kinetics and Reaction Dynamics}.
\newblock McGrw-Hill Higher Education, New York, NY, 2001.

\bibitem{HSWK}
J.H. Hunter, M.T. Sandford, R.W. Whitaker, and R.I. Klein.
\newblock Star formation in colliding gas-flows.
\newblock {\em Astrophysical Journal}, 305(1):309--332, 1986.

\bibitem{JLPH}
J.~Jung, R.W. Lyczkowski, C.~Panchal, and A.~Hassanein.
\newblock Multiphase hemodynamic simulation of pulsatile flow in a coronary
  artery.
\newblock {\em Journal of Biomechanics}, 39(11):2064--2073, 2006.

\bibitem{KS}
A.~V. Kazhikhov and V.~V. Shelukhin.
\newblock Unique global solution with respect to time of initial-boundary value
  problems for one-dimensional equations of a viscous gas.
\newblock {\em Prikl. Mat. Meh.}, 41(2):282--291, 1977.

\bibitem{Kazhikhov}
A.~V. Ka{\v{z}}ihov.
\newblock Solvability of the initial-boundary value problem for the equations
  of the motion of an inhomogeneous viscous incompressible fluid.
\newblock {\em Dokl. Akad. Nauk SSSR}, 216:1008--1010, 1974.

\bibitem{SEV}
Shui L., J.C.T. Eijkel, and A.~van~den Berg.
\newblock Multiphase flow in microfluidic systems - control and applications of
  droplets and interfaces.
\newblock {\em Advances in Colloid and Interface Science}, (133):35--49, 2007.

\bibitem{LSU}
O.~A. Ladyzenskaja, V.~A. Solonnikov, and N.~N. Ural$'$ceva.
\newblock {\em Lineinye i kvazilineinye uravneniya parabolicheskogo tipa}.
\newblock Izdat. ``Nauka'', Moscow, 1968.

\bibitem{LKBJS}
G.~Lemon, J.R. King, H.M. Byrne, O.E. Jensen, and K.M. Shakesheff.
\newblock Mathematical modelling of engineered tissue growth using a multiphase
  porous flow mixture theory.
\newblock {\em J. Math. Biol.}, 52(5):571--594, 2006.

\bibitem{PLL1}
P-L. Lions.
\newblock {\em Mathematical topics in fluid mechanics. {V}ol. 1}, volume~3 of
  {\em Oxford Lecture Series in Mathematics and its Applications}.
\newblock The Clarendon Press Oxford University Press, New York, 1996.
\newblock Incompressible models, Oxford Science Publications.

\bibitem{PLL2}
P-L. Lions.
\newblock {\em Mathematical topics in fluid mechanics. {V}ol. 2}, volume~10 of
  {\em Oxford Lecture Series in Mathematics and its Applications}.
\newblock The Clarendon Press Oxford University Press, New York, 1998.
\newblock Compressible models, Oxford Science Publications.

\bibitem{LWZ}
S.~Liu, F.~Wang, and H.~Zhao.
\newblock Global existence and asymptotics of solutions of the
  {C}ahn-{H}illiard equation.
\newblock {\em J. Differential Equations}, 238(2):426--469, 2007.

\bibitem{MR1}
J.~M{\'a}lek, G.~Mingione, and J.~Star{\'a}.
\newblock Fluids with pressure dependent viscosity: partial regularity of
  steady flows.
\newblock In {\em EQUADIFF 2003}, pages 380--385. World Sci. Publ., Hackensack,
  NJ, 2005.

\bibitem{MR2}
J.~M{\'a}lek and K.~R. Rajagopal.
\newblock Incompressible rate type fluids with pressure and shear-rate
  dependent material moduli.
\newblock {\em Nonlinear Anal. Real World Appl.}, 8(1):156--164, 2007.

\bibitem{MN1}
A.~Matsumura and T.~Nishida.
\newblock The initial value problem for the equations of motion of compressible
  viscous and heat-conductive fluids.
\newblock {\em Proc. Japan Acad. Ser. A Math. Sci.}, 55(9):337--342, 1979.

\bibitem{MN2}
A.~Matsumura and T.~Nishida.
\newblock The initial value problem for the equations of motion of viscous and
  heat-conductive gases.
\newblock {\em J. Math. Kyoto Univ.}, 20(1):67--104, 1980.

\bibitem{MN3}
A.~Matsumura and T.~Nishida.
\newblock Initial-boundary value problems for the equations of compressible
  viscous and heat-conductive fluid.
\newblock In {\em Nonlinear partial differential equations in applied science
  (Tokyo, 1982)}, volume~81 of {\em North-Holland Math. Stud.}, pages 153--170.
  North-Holland, Amsterdam, 1983.

\bibitem{Mayer}
A.B.R. Mayer.
\newblock Colloidal metal nanoparticles dispersed in amphiphilic polymers.
\newblock {\em Polymers for Advanced Technologies}, 12(1--2):96--106, 2001.

\bibitem{MV1}
A.~Mellet and A.~Vasseur.
\newblock Existence and uniqueness of global strong solutions for
  one-dimensional compressible navier-stokes equations.
\newblock {\em In Press}, 2007.

\bibitem{MV2}
A.~Mellet and A.~Vasseur.
\newblock Global weak solutions for a
  {V}lasov-{F}okker-{P}lanck/{N}avier-{S}tokes system of equations.
\newblock {\em Math. Models Methods Appl. Sci.}, 17(7):1039--1063, 2007.

\bibitem{MV3}
A.~Mellet and A.~Vasseur.
\newblock On the barotropic compressible {N}avier-{S}tokes equations.
\newblock {\em Comm. Partial Differential Equations}, 32(1-3):431--452, 2007.

\bibitem{MRFJ}
F.~Miniati, D.S. Ryu, A.~Ferrara, and T.W. Jones.
\newblock Magnetohydrodynamics of cloud collisions in a multiphase interstellar
  medium.
\newblock {\em Astrophysical Journal}, 510(2):726--746, 1979.

\bibitem{MS}
H.~Moehwald and D.G. Shchukin.
\newblock Sonochemical nanosynthesis at the engineered interface of a
  cavitation microbubble.
\newblock {\em Physical Chemistry Chemical Physics}, 8(30):3496--3506, 2006.

\bibitem{NP1}
A.~Nouri and F.~Poupaud.
\newblock An existence theorem for the multifluid {N}avier-{S}tokes problem.
\newblock {\em J. Differential Equations}, 122(1):71--88, 1995.

\bibitem{NP2}
A.~Nouri, F.~Poupaud, and Y.~Demay.
\newblock An existence theorem for the multi-fluid {S}tokes problem.
\newblock {\em Quart. Appl. Math.}, 55(3):421--435, 1997.

\bibitem{NS}
A.~Novotn{\'y} and I.~Stra{\v{s}}kraba.
\newblock {\em Introduction to the mathematical theory of compressible flow},
  volume~27 of {\em Oxford Lecture Series in Mathematics and its Applications}.
\newblock Oxford University Press, Oxford, 2004.

\bibitem{OCENS}
T.E. Ongaro, C.~Cavazzoni, G.~Erbacci, A.~Neri, and M.~V. Salvetti.
\newblock A parallel multiphase flow code for the 3d simulation of explosive
  volcanic eruptions.
\newblock {\em Parallel Computing}, 33(7--8):541--560, 2007.

\bibitem{OD}
S.~O'Sullivan and T.P. Downes.
\newblock An explicit scheme for multifluid magnetohydrodynamics.
\newblock {\em Monthly Notices of the Royal Astronomical Society},
  366(4):1329--1336, 2006.

\bibitem{PL}
Shih-I Pai and Shijun Luo.
\newblock {\em Theoretical and computational dynamics of a compressible flow}.
\newblock Beijing: Science Press, New York, NY, 1991.

\bibitem{PW}
U.~Pasaogullari and C.Y. Wang.
\newblock Liquid water transport in gas diffusion layer of polymer electrolyte
  fuel cells.
\newblock {\em Journal of the Electrochemical Society}, 151(3):A399--A406,
  2004.

\bibitem{Pedlosky}
J.~Pedlosky.
\newblock {\em Geophysical fluid dynamics}, volume 2nd Edition.
\newblock Springer-Verlag New York Inc., New York, NY, 1987.

\bibitem{PSW}
K.~Promislow, J.~Stockie, and B.~Wetton.
\newblock A sharp interface reduction for multiphase transport in a porous fuel
  cell electrode.
\newblock {\em Proc. R. Soc. Lond. Ser. A Math. Phys. Eng. Sci.},
  462(2067):789--816, 2006.

\bibitem{SIG}
Safran S., T.~Kuhl, J.~Israelachvili, and G.~Hed.
\newblock Polymer induced membrane contraction, phase separation, and fusion
  via marangoni flow.
\newblock {\em Biophysical Journal}, 81(2):659--666, 2001.

\bibitem{Serre2}
D.~Serre.
\newblock Solutions faibles globales des \'equations de {N}avier-{S}tokes pour
  un fluide compressible.
\newblock {\em C. R. Acad. Sci. Paris S\'er. I Math.}, 303(13):639--642, 1986.

\bibitem{Serre1}
D.~Serre.
\newblock Sur l'\'equation monodimensionnelle d'un fluide visqueux,
  compressible et conducteur de chaleur.
\newblock {\em C. R. Acad. Sci. Paris S\'er. I Math.}, 303(14):703--706, 1986.

\bibitem{Sh1}
V.~V. Shelukhin.
\newblock Motion with a contact discontinuity in a viscous heat conducting gas.
\newblock {\em Dinamika Sploshn. Sredy}, (57):131--152, 1982.

\bibitem{Sh2}
V.~V. Shelukhin.
\newblock Evolution of a contact discontinuity in the barotropic flow of a
  viscous gas.
\newblock {\em Prikl. Mat. Mekh.}, 47(5):870--872, 1983.

\bibitem{Sh4}
V.~V. Shelukhin.
\newblock Boundary value problems for equations of a barotropic viscous gas
  with nonnegative initial density.
\newblock {\em Dinamika Sploshn. Sredy}, (74):108--125, 162--163, 1986.

\bibitem{S1}
V.~A. Solonnikov.
\newblock The solvability of the initial-boundary value problem for the
  equations of motion of a viscous compressible fluid.
\newblock {\em Zap. Nau\v cn. Sem. Leningrad. Otdel. Mat. Inst. Steklov.
  (LOMI)}, 56:128--142, 197, 1976.
\newblock Investigations on linear operators and theory of functions, VI.

\bibitem{S2}
V.~A. Solonnikov.
\newblock Unsteady flow of a finite mass of a fluid bounded by a free surface.
\newblock {\em Zap. Nauchn. Sem. Leningrad. Otdel. Mat. Inst. Steklov. (LOMI)},
  152(Kraev. Zadachi Mat. Fiz. i Smezhnye Vopr. Teor. Funktsii18):137--157,
  183--184, 1986.

\bibitem{S3}
V.~A. Solonnikov.
\newblock On a nonstationary motion of a finite mass of a liquid bounded by a
  free surface.
\newblock In {\em Differential equations (Xanthi, 1987)}, volume 118 of {\em
  Lecture Notes in Pure and Appl. Math.}, pages 647--653. Dekker, New York,
  1989.

\bibitem{S4}
V.~A. Solonnikov.
\newblock Unsteady motions of a finite isolated mass of a self-gravitating
  fluid.
\newblock {\em Algebra i Analiz}, 1(1):207--249, 1989.

\bibitem{S5}
V.~A. Solonnikov.
\newblock Unsteady motions of a finite isolated mass of a self-gravitating
  fluid.
\newblock {\em Algebra i Analiz}, 1(1):207--249, 1989.

\bibitem{ST2}
V.~A. Solonnikov and A.~Tani.
\newblock A problem with a free boundary for {N}avier-{S}tokes equations for a
  compressible fluid in the presence of surface tension.
\newblock {\em Zap. Nauchn. Sem. Leningrad. Otdel. Mat. Inst. Steklov. (LOMI)},
  182(Kraev. Zadachi Mat. Fiz. i Smezh. Voprosy Teor. Funktsii. 21):142--148,
  173--174, 1990.

\bibitem{ST}
V.~A. Solonnikov and A.~Tani.
\newblock Evolution free boundary problem for equations of motion of viscous
  compressible barotropic liquid.
\newblock In {\em The Navier-Stokes equations II---theory and numerical methods
  (Oberwolfach, 1991)}, volume 1530 of {\em Lecture Notes in Math.}, pages
  30--55. Springer, Berlin, 1992.

\bibitem{T1}
N.~Tanaka.
\newblock Global existence of two phase nonhomogeneous viscous incompressible
  fluid flow.
\newblock {\em Comm. Partial Differential Equations}, 18(1-2):41--81, 1993.

\bibitem{T2}
A.~Tani and N.~Tanaka.
\newblock Large-time existence of surface waves in incompressible viscous
  fluids with or without surface tension.
\newblock {\em Arch. Rational Mech. Anal.}, 130(4):303--314, 1995.

\bibitem{TGS}
P.~Tartaj, T.~Gonzalez-Carreno, and C.J. Serna.
\newblock Magnetic behavior of gamma-fe2o3 nanocrystals dispersed in colloidal
  silica particles.
\newblock {\em Journal of Physical Chemistry B}, 107(1):20--24, 1993.

\bibitem{KT}
K.~Trivisa.
\newblock Global existence and asymptotic analysis of solutions to a model for
  the dynamic combustion of compressible fluids.
\newblock {\em Discrete Contin. Dyn. Syst.}, (suppl.):852--863, 2003.
\newblock Dynamical systems and differential equations (Wilmington, NC, 2002).

\bibitem{VZ}
A.~Valli and W.M. Zajaczkowski.
\newblock Navier-{S}tokes equations for compressible fluids: global existence
  and qualitative properties of the solutions in the general case.
\newblock {\em Comm. Math. Phys.}, 103(2):259--296, 1986.

\bibitem{Vallis}
G.~Vallis.
\newblock {\em Atmospheric and oceanic fluid dynamics : fundamentals and
  large-scale circulation}, volume 2nd Edition.
\newblock Cambridge University Press, New York, NY, 2006.

\bibitem{Williams}
F.A. Williams.
\newblock {\em Combustion Theory}, volume Second Edition of {\em Combustion
  Science and Engineering Series}.
\newblock The Benjamin/Cummings Publishing Company, Inc., Menlo Park,
  California, 1985.

\bibitem{WMW}
C.Y. Wong, J.A. Maruhn, and T.A. Welton.
\newblock Dynamics of nuclear fluids. i. foundations.
\newblock {\em Nucl. Phys.}, A253:469--489, 1975.

\bibitem{Wyatt}
R.E. Wyatt.
\newblock {\em Quantum dynamics with trajectories}, volume~28 of {\em
  Interdisciplinary Applied Mathematics}.
\newblock Springer-Verlag, New York, 2005.
\newblock Introduction to quantum hydrodynamics, With contributions by Corey J.
  Trahan.

\bibitem{YZ}
T.~Yang and C.~Zhu.
\newblock Compressible {N}avier-{S}tokes equations with degenerate viscosity
  coefficient and vacuum.
\newblock {\em Comm. Math. Phys.}, 230(2):329--363, 2002.

\bibitem{Youngs}
D.L. Youngs.
\newblock Numerical-simulation of turbulent mixing by rayleigh-taylor
  instability.
\newblock {\em Physica D}, 12(1--3):32--44, 1984.

\bibitem{Z1}
A.~A. Zlotnik.
\newblock Weak solutions of the equations of motion of a viscous compressible
  reacting binary mixture: uniqueness and {L}ipschitz-continuous dependence on
  data.
\newblock {\em Mat. Zametki}, 75(2):307--311, 2004.

\bibitem{ZD}
A.~A. Zlotnik and B.~Dyukome.
\newblock Stabilization of one-dimensional flows of a radiative and a reactive
  viscous gas for a general rate of reaction.
\newblock {\em Dokl. Akad. Nauk}, 403(6):731--736, 2005.

\end{thebibliography}
\end{document}